\documentclass[lettersize,journal]{IEEEtran}
\usepackage{amsmath,amsfonts}
\usepackage{amsthm}
\usepackage{array}
\usepackage{textcomp}
\usepackage{stfloats}
\usepackage{url}
\usepackage{verbatim}
\usepackage{graphicx}
\usepackage{comment}
\usepackage{physics}
\usepackage{algorithm}
\usepackage{algpseudocode}
\usepackage{subfigure}
\usepackage{xcolor}
\usepackage{soul,color}

\usepackage{amsmath,amsfonts,bm}

\def\eqref#1{Eq.~(\ref{#1})}

\def\1{\bm{1}}

\DeclareMathAlphabet{\mathsfit}{\encodingdefault}{\sfdefault}{m}{sl}
\SetMathAlphabet{\mathsfit}{bold}{\encodingdefault}{\sfdefault}{bx}{n}

\newtheorem{theorem}{Theorem}
\newtheorem{proposition}[theorem]{Proposition}%
\newtheorem{lemma}{Lemma}
\newtheorem{remark}{Remark}

\def\BibTeX{{\rm B\kern-.05em{\sc i\kern-.025em b}\kern-.08em
    T\kern-.1667em\lower.7ex\hbox{E}\kern-.125emX}}
\usepackage{balance}
\begin{document}
\title{Trust Region Bayesian Optimization of Annealing Schedules on a Quantum Annealer}
\author{Seon-Geun Jeong, Mai Dinh Cong, Minh-Duong~Nguyen, Xuan Tung~Nguyen, \\Quoc-Viet Pham, and Won-Joo Hwang

\thanks{Seon-Geun Jeong and Mai Dinh Cong are with the Department of Information Convergence Engineering, Pusan National University, Busan 46241, South Korea (e-mail: wjdtjsrms11@pusan.ac.kr, cong.md@pusan.ac.kr).}%
\thanks{Minh-Duong Nguyen is with the Department of Intelligent Computing and Data Science, VinUniversity, Hanoi, Vietnam (e-mail: duong.nm2@vinuni.edu.vn).}
\thanks{Xuan Tung Nguyen is with the Faculty of Interdisciplinary Digital Technology,
PHENIKAA University, Yen Nghia, Ha Dong, Hanoi 10000, Viet Nam (email: tung.nguyenxuan@phenikaa-uni.edu.vn).}
\thanks{Quoc-Viet Pham is with the School of Computer Science and Statistics, Trinity College Dublin, The University of Dublin, Dublin 2, D02 PN40, Ireland (e-mail: viet.pham@tcd.ie).}
\thanks{Won-Joo Hwang (corresponding author) is with the School of Computer Science and Engineering, Center for Artificial Intelligence Research, Pusan National University, Busan 46241, South Korea (e-mail: wjhwang@pusan.ac.kr).}}%


\maketitle

\begin{abstract}
Quantum annealing (QA) is a practical model of adiabatic quantum computation, already realized on hardware and considered promising for combinatorial optimization. However, its performance is critically dependent on the annealing schedule due to hardware decoherence and noise. Designing schedules that account for such limitations remains a significant challenge. We propose a trust region Bayesian optimization (TuRBO) framework that jointly tunes annealing time and Fourier-parameterized schedules. Given a fixed embedding on a quantum processing unit (QPU), the framework employs Gaussian process surrogates with expected improvement to balance exploration and exploitation, while trust region updates refine the search around promising candidates. The framework further incorporates mechanisms to manage QPU runtime and enforce feasibility under hardware constraints efficiently. Simulation studies demonstrate that TuRBO consistently identifies schedules that outperform random and greedy search in terms of energy, feasible solution probability, and chain break fraction. These results highlight TuRBO as a resource-efficient and scalable strategy for annealing schedule design, offering improved QA performance in noisy intermediate-scale quantum regimes and extensibility to industrial optimization tasks.
\end{abstract}


\section{Introduction}
\IEEEPARstart{A}{diabatic} quantum computation (AQC) provides a theoretical framework for solving optimization problems by evolving a quantum system from the ground state of an initial Hamiltonian to that of a problem Hamiltonian under slow, continuous changes~\cite{albash2018adiabatic}. According to the adiabatic theorem~\cite{kato1950adiabatic}, if the evolution is sufficiently slow relative to the minimum spectral gap, i.e., the smallest energy difference between the ground state and the first excited state, the system remains in the instantaneous ground state, yielding the desired solution. While AQC is polynomially equivalent to the circuit model, it offers an alternative perspective on quantum algorithms and has inspired a range of heuristic methods for practical implementations~\cite{yang2020optimizing}.

Quantum annealing (QA) represents a practical realization of AQC and has been implemented on hardware platforms such as the D-Wave systems~\cite{dwave}. In QA, the annealing schedule, which governs the interpolation between the initial and problem Hamiltonians, plays a critical role in determining performance~\cite{yang2020optimizing, chen2022optimizing}. Similarly, in gate-based approaches such as the quantum approximate optimization algorithm, the variational parameters can be interpreted as a discretized schedule~\cite{brady2021optimal}. Both paradigms underscore the importance of schedule design, as the choice of path significantly affects success probability, solution quality, and robustness to noise.

However, designing annealing schedules that perform well on real devices remains a major challenge. Hardware limitations, such as decoherence, control noise, and finite runtime budgets, often prevent the use of schedules derived purely from theoretical considerations~\cite{chen2022optimizing}. Prior studies have explored schedule optimization using reinforcement learning and Monte Carlo tree search (MCTS)~\cite{chen2022optimizing}, or gradient-free learning approaches such as differential evolution~\cite{yang2020optimizing}. While these efforts have demonstrated promising results, they are largely conducted in simulation and often do not fully incorporate device-level constraints. This gap highlights the need for systematic and resource-efficient approaches that directly account for practical hardware conditions.

Motivated by the effectiveness of Bayesian optimization (BO) for global optimization and hyperparameter tuning in machine learning~\cite{wu2019hyperparameter}, recent works have applied the standard BO to QA schedule design, showing that it can discover high-quality paths with minimal user intervention~\cite{finvzgar2024designing}. However, standard BO is most effective when the number of tunable parameters is small, and it often scales poorly when the dimensionality grows or when the objective is noisy and expensive to evaluate~\cite{eriksson2019scalable}. This limitation becomes particularly critical in annealing schedule design, where Fourier parameterization introduces multiple coefficients in addition to the annealing time, expanding the search space beyond the practical capacity of standard BO~\cite{yang2020optimizing, chen2022optimizing}. To address this, we adopt trust region Bayesian optimization (TuRBO)~\cite{eriksson2019scalable}, which restricts exploration to adaptive local regions and thereby offers scalability, robustness to noise, and efficient use of limited quantum processing unit (QPU) budgets.

In our approach, the annealing path is explicitly expressed in a truncated Fourier basis, and TuRBO is employed to jointly optimize the total runtime and the Fourier coefficients. A Gaussian process (GP) surrogate with expected improvement balances exploration and exploitation, while a trust region (TR) mechanism refines the search around promising candidates. To ensure practicality on QPUs, the framework further integrates adaptive readout scaling, runtime budget guards, and feasibility checks. This design enables schedule optimization directly under noisy intermediate-scale quantum (NISQ) conditions without excessive sampling overhead.    

We further demonstrate the effectiveness of our method on the traveling salesman problem (TSP) benchmark, a canonical NP-hard problem, and validate the framework on the D-Wave quantum annealer. To the best of our knowledge, this represents the first application of TuRBO for Fourier-parameterized schedule design executed on a practical QA device. Experimental results show that our optimized schedules consistently outperform random search (RS) and greedy search (GS) baselines, yielding higher success probability, lower chain break fraction (CBF), and improved energy quality. These findings highlight that the proposed framework not only bridges the gap between theory and practice but also provides a robust and scalable path toward enhancing QA performance on real hardware.
The main contributions of this study are summarized as follows:
\begin{itemize}
    \item We propose a TuRBO-based framework that jointly optimizes annealing time and Fourier-parameterized schedules on a real quantum annealer.
    \item Hardware-aware optimization mechanisms are designed, including adaptive readout scaling, runtime budget control, and feasibility checks, to ensure efficient and reliable execution on the real quantum annealers.
    \item Extensive empirical validation on TSP instances demonstrates substantial improvements in success probability, chain-break fraction, and solution quality compared to RS and GS baselines.
    \item We empirically compare the proposed QA framework with classical baselines such as simulated annealing (SA), genetic algorithm (GA), and CPLEX optimization, demonstrating clear runtime advantages.
    \item We incorporate the integrated control error (ICE) model into our analysis to assess the scalability limitations of QA under realistic noise, providing experimental validation of its impact.
    \item To the best of our knowledge, this study presents the first real-hardware demonstration of TuRBO-driven Fourier-based schedule optimization on the D-Wave quantum annealer, bridging theoretical advances and practical deployment.
\end{itemize}

The remainder of this paper is organized as follows. 
Section~\ref{sec:related} reviews related work on annealing schedule design. 
Section~\ref{sec:pre} presents the preliminaries of AQC, QA, and annealing schedule. 
Section~\ref{sec:pro} formulates the TSP, derives its quadratic unconstrained binary optimization (QUBO) representation, and defines the annealing schedule optimization.
Section~\ref{sec:met} details the proposed TuRBO framework, including the surrogate model, acquisition strategy, and QPU-aware mechanisms. 
Section~\ref{sec:Exp} reports experimental validation on the TSP benchmark and real-hardware implementation on the D-Wave quantum annealer. 
Section~\ref{sec:discussion} discusses scalability limitations such as noise effects and highlights directions for practical deployment.
Finally, Section~\ref{sec:Con} concludes the paper and outlines future research directions.

\section{Related Works}
\label{sec:related}

\subsection{Annealing Schedule Optimization}
The performance of QA is strongly affected by the annealing schedule, which controls the interpolation between the initial and problem Hamiltonians. Brif \textit{et al.}~\cite{brif2014exploring} applied quantum optimal control theory to design adiabatic trajectories with composite objectives, showing that multiple control functions can enlarge spectral gaps and improve fidelity. Yang \textit{et al.}~\cite{yang2020optimizing} further employed a differential evolution algorithm combined with the chopped random basis technique, where schedules are expressed in truncated Fourier series, to identify adiabatic pathways formulated as a multi-objective optimization problem. Both studies highlight the potential of control-theoretic methods but remain limited to small-scale numerical simulations without hardware validation. Chen \textit{et al.}~\cite{chen2022optimizing} introduced a reinforcement learning approach using MCTS with neural networks, demonstrating the ability to discover nontrivial schedules for satisfiability problems. Fin\v{z}gar \textit{et al.}~\cite{finvzgar2024designing} applied BO to schedule design and reported significant fidelity improvements over linear baselines. However, their analysis was limited to theoretical settings without incorporating hardware-level constraints, and the annealing time was fixed rather than jointly optimized with the schedule.
Khezri \textit{et al.}~\cite{khezri2022customized} investigated customized schedules in superconducting flux-qubit circuits by mapping circuit-level flux biases to effective Pauli schedules, providing a simulation-based study of hardware constraints. Their work offered device-level customization but did not demonstrate algorithmic optimization or validation on real hardware.

\subsection{Quantum Annealing Applications}
Le \textit{et al.}~\cite{le2023quantum} investigated the selective TSP with a novel QUBO formulation on the D-Wave 2000Q. The annealing was performed under fixed schedules with predetermined durations, and no justification was provided for the choice of time or schedule shape.  
Ikeda \textit{et al.}~\cite{ikeda2019application} addressed the nurse scheduling problem on the D-Wave 2000Q and explored reverse annealing as a means of refining solutions. While this introduced a reverse annealing mechanism, the method was applied heuristically without a systematic rationale for schedule parameters.  
Carugno \textit{et al.}~\cite{carugno2022evaluating} studied the job shop scheduling problem on the D-Wave Advantage, focusing on embedding and chain-break mitigation. Their experiments highlighted the influence of annealing duration and the use of reverse annealing, yet these controls were selected heuristically rather than optimized.  
Pérez-Armas \textit{et al.}~\cite{perez2024solving} examined the resource-constrained project scheduling problem on the D-Wave Advantage, introducing benchmarking metrics such as time-to-target and Q-score. Although annealing time and pause effects were varied, the schedules were still manually chosen without principled design.  
Authors in~\cite{10907925, ohyama2021intelligent} formulated the wireless communication problem as a QUBO and evaluated it using simulations and hybrid QA solvers. While acknowledging that annealing schedule choice can affect performance, the study employed default settings without investigating alternative shapes or durations.  

In contrast to these prior works, this study directly addresses the problem of annealing schedule design under realistic hardware constraints. Whereas earlier approaches either fixed the annealing time, adopted default schedules, or applied mechanisms such as reverse annealing in a heuristic manner, our framework introduces a principled optimization strategy. Specifically, we employ TuRBO to jointly tune the annealing time and Fourier-parameterized schedule, enabling systematic exploration of nontrivial annealing shapes beyond linear or manually selected forms. By incorporating QPU-aware mechanisms such as adaptive readout scaling, runtime budget guards, and feasibility checks, the proposed method is tailored for NISQ devices. Furthermore, to the best of our knowledge, this work presents the first real-hardware demonstration of TuRBO-driven schedule optimization on the D-Wave quantum annealer, validated on TSP instances. These contributions establish a novel and practical pathway for bridging theoretical schedule optimization and applied QA on current hardware.

\section{Preliminaries}\label{sec:pre}
\subsection{Adiabatic Quantum Computation}
AQC is a model of quantum computing grounded in the adiabatic theorem~\cite{kato1950adiabatic}. 
AQC encodes a computational problem into a problem Hamiltonian $H_P$, whose ground state corresponds to the optimal solution. 
The system is initialized in the ground state of an initial Hamiltonian $H_0$, and the total Hamiltonian is gradually transformed as
\begin{equation}
    H(t) = (1-s(t)) H_0 + s(t) H_P, \quad s(t) \in [0,1],
\end{equation}
where $s(t)$ is a monotonic scheduling function of time $t$. 

According to the adiabatic theorem, the success of the evolution depends on the minimum spectral gap as follows:
\begin{equation}
    \Delta_{\min} = \min_{s \in [0,1]} \left( E_1(s) - E_0(s) \right),
\end{equation}
where $E_0(s)$ and $E_1(s)$ denote the instantaneous ground and first excited state energies. 
The adiabatic condition can be expressed as
\begin{equation}
    T \gg \frac{\max_s \left| \langle \psi_1(s) \big| \tfrac{dH(s)}{ds} \big| \psi_0(s) \rangle \right|}{\Delta_{\min}^2},
    \label{eq:adiabatic-runtime}
\end{equation}
where $T$ is the total runtime and $\ket{\psi_0(s)}$, $\ket{\psi_1(s)}$ are the instantaneous eigenstates of $H(s)$. 
This condition highlights that smaller spectral gaps require slower evolution (larger $T$) to ensure the system remains in its ground state throughout the computation. 

\begin{remark}
In practice, the numerator term in~\eqref{eq:adiabatic-runtime} can be upper-bounded by a polynomial in the problem size $N$, so the overall complexity is determined primarily by the scaling of the minimum spectral gap $\Delta_{\min}$. When $\Delta_{\min}$ decreases only as $1/\mathrm{poly}(N)$, the required runtime scales polynomially with $N$ and AQC can, in principle, be efficient. By contrast, if $\Delta_{\min}$ closes exponentially, for example $\Delta_{\min} \sim e^{-\alpha N}$ with $\alpha>0$, the runtime becomes exponential, implying that AQC offers no quantum speedup~\cite{albash2018adiabatic}. This underscores the critical role of spectral gap analysis not only in evaluating the theoretical power of AQC, but also in guiding the design of annealing schedules, since regions with small gaps require slower evolution while regions with larger gaps can be traversed more quickly.
\end{remark}

\subsection{Quantum Annealing}
QA is a practical realization of AQC designed for solving combinatorial optimization problems~\cite{kadowaki1998quantum}. 
In QA, the target problem is formulated either as an Ising Hamiltonian
\begin{equation}
    H_P = \sum_{i=1}^N h_i \sigma_i^z + \sum_{\substack{i,j=1 \\ i<j}}^{N} J_{ij} \sigma_i^z \sigma_j^z,
\end{equation}
or equivalently as a QUBO model
\begin{equation}
    f_Q(\boldsymbol{x}) = \boldsymbol{x}^T Q \boldsymbol{x}, \quad \boldsymbol{x} \in \{0,1\}^N.
\end{equation}

\subsubsection{Adiabatic Evolution}
The computation begins with an initial Hamiltonian $H_0$, usually chosen as the transverse-field Hamiltonian~\cite{farhi2002quantum}:
\begin{equation}
    H_0 = -\sum_{i=1}^N \sigma_i^x,
\end{equation}
whose ground state is an easily prepared uniform superposition of all computational basis states. 
The system then evolves according to a time-dependent Hamiltonian
\begin{equation}
    H(s) = A(s) H_0 + B(s) H_P, \quad s \in [0,1],
\end{equation}
where $A(s)$ decreases from its maximum value to zero, while $B(s)$ increases from zero to its maximum value. 
This interpolation ensures that $H(0)=H_0$ and $H(1)=H_P$, ideally transforming the system from the ground state of $H_0$ to the ground state of $H_P$.

\subsubsection{Annealing Schedules}
The function $s(t)$, known as the annealing schedule, determines how the interpolation parameter evolves over the annealing time $T$. 
A common choice is the linear schedule,
\begin{equation}
    s(t) = \frac{t}{T}, \quad t \in [0,T],
\end{equation}
but more general non-linear schedules can be employed to slow down the evolution near small spectral gaps or to accelerate it in regions where the energy spectrum is wide. 
Variants such as pauses, quenches, or reverse annealing have also been introduced in practice, but most applications still rely on fixed or heuristic schedules without principled optimization.

\subsubsection{D-Wave Quantum Annealer}
Commercial implementations of QA are provided by D-Wave Systems~\cite{dwave}, which employ superconducting flux qubits arranged in sparse topologies such as Chimera, Pegasus, and Zephyr. 
Problems expressed in QUBO or Ising form are first converted into a logical graph and then mapped onto the hardware using minor embedding, which introduces chains of physical qubits to represent logical variables. 
The annealing process is governed by the schedules $A(s)$ and $B(s)$, followed by a measurement of qubit states to obtain classical solutions. 
Figure~\ref{fig:dwave-framework} illustrates the overall workflow of the D-Wave QA process, from QUBO formulation to hardware embedding, annealing, and final readout.

\begin{figure}[h]
    \centering
    \includegraphics[width=0.5\textwidth]{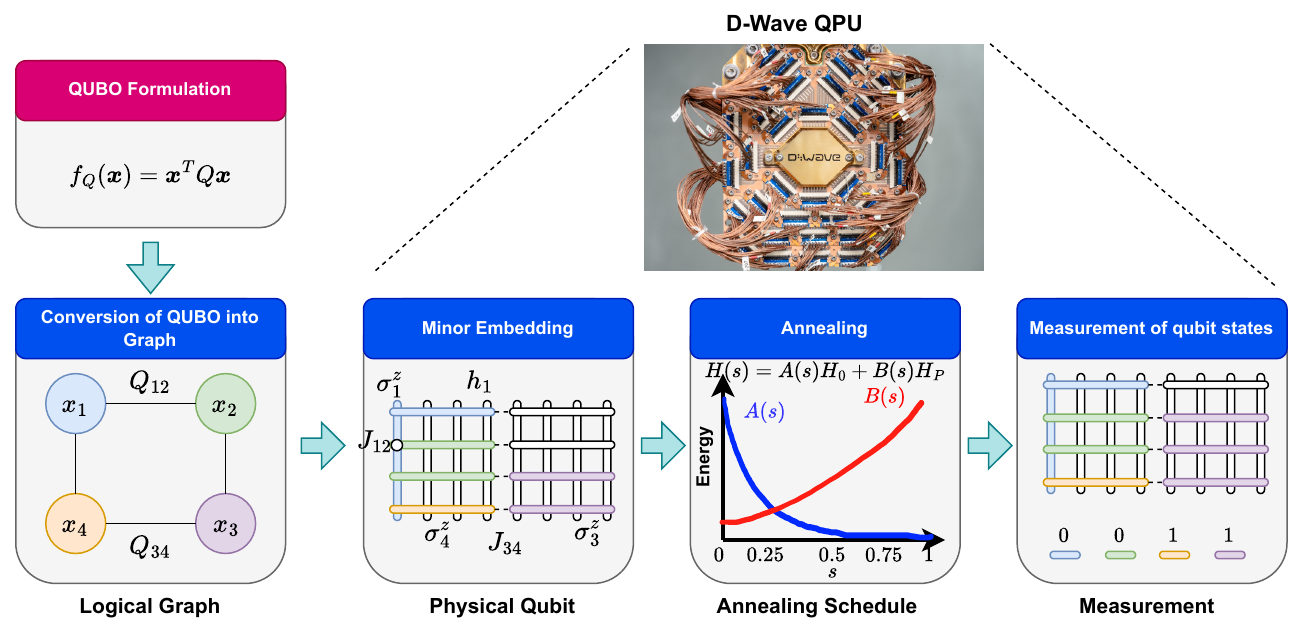}
    \caption{D-Wave Quantum Annealing Framework. The process begins with casting the problem into a QUBO form, followed by its translation into a logical graph. This graph is then embedded onto the hardware topology through minor embedding, after which the annealing dynamics defined by $A(s)$ and $B(s)$ are executed. Finally, the qubits are measured to obtain classical bitstring solutions.}
    \label{fig:dwave-framework}
\end{figure}

\section{Problem Statement}\label{sec:pro}

\subsection{Traveling Salesman Problem}
The TSP is a well-known NP-hard optimization problem that requires finding the shortest route visiting each city exactly once and returning to the starting city. Let $N$ be the number of cities, and let $d_{i,j}$ denote the distance (or cost) of traveling from city $i$ to city $j$.  
We introduce binary decision variables $x_{i,k}$ defined as
\begin{equation}
x_{i,k} =
\begin{cases}
1 & \text{if city $i$ is visited at step $k$,} \\
0 & \text{otherwise,}
\end{cases}
\end{equation}
where $i,k \in \{1,\dots,N\}$. The total number of binary variables is therefore $N^2$. These variables can be arranged into a binary vector
\begin{equation}
\boldsymbol{x} = [x_{1,1}, x_{1,2}, \dots, x_{1,N}, x_{2,1}, \dots, x_{N,N}]^\top \in \{0,1\}^{N^2},
\end{equation}
which encodes a complete candidate tour configuration. Each valid tour corresponds to a unique binary vector \(\boldsymbol{x}\) satisfying the constraints below.

\subsubsection{Constraints} To ensure a valid tour, the binary variables must satisfy two types of constraints. 
First, each city must be visited exactly once:
\begin{equation}
\sum\nolimits_{k=1}^N x_{i,k} = 1, \quad \forall i \in \{1,\dots,N\}.
\label{eq:tsp-city}
\end{equation}
Second, each step of the tour must correspond to exactly one city:
\begin{equation}
\sum\nolimits_{i=1}^N x_{i,k} = 1, \quad \forall k \in \{1,\dots,N\}.
\label{eq:tsp-step}
\end{equation}

\subsubsection{Objective function} 
The total tour length is expressed as
\begin{equation}
\min_{\boldsymbol{x}} \sum_{i=1}^N \sum_{j=1}^N \sum_{k=1}^N d_{i,j}\, x_{i,k} \, x_{j,(k \bmod N) + 1},
\end{equation}
where $d_{i,j}$ denotes the distance between city $i$ and city $j$. 
The index $(k \bmod N) + 1$ ensures that when $k = N$, the route returns to the start city, thereby forming a closed cycle. These constraints and the objective function together define the classical TSP. In the next subsection, we show how to encode them into a QUBO form suitable for QA.

\subsection{Quadratic Unconstrained Binary Optimization}
To solve the TSP on a quantum annealer, the constrained formulation must be converted into a QUBO problem. 
This is achieved by introducing penalty terms for the constraints and combining them with the objective function. 
Let $\lambda$ be a large positive constant to weight the constraint violations.

\subsubsection{Penalty terms} The tour validity constraints from~\eqref{eq:tsp-city}-\eqref{eq:tsp-step} can be encoded as
\begin{equation}
H_c(\boldsymbol{x}) = 
\lambda \sum_{i=1}^{N} \left(1 - \sum_{k=1}^{N} x_{i,k}\right)^2
+ \lambda \sum_{k=1}^{N} \left(1 - \sum_{i=1}^{N} x_{i,k}\right)^2.
\end{equation}
By construction, this penalty term vanishes if and only if the constraints are satisfied, meaning that each city is visited exactly once and each tour step corresponds to one city.

\subsubsection{Objective term} 
The cost function introduced in the TSP formulation can be directly written as
\begin{equation}
H_{\text{obj}}(\boldsymbol{x}) = 
\sum_{i=1}^{N} \sum_{j=1}^{N} \sum_{k=1}^{N} d_{i,j} x_{i,k} x_{j, (k \bmod N) + 1}.
\end{equation}
This term corresponds to the total tour length and favors solutions with shorter routes.

\subsubsection{Final QUBO form} The complete QUBO objective is obtained by combining the penalties and the cost:
\begin{equation}
\label{eq:tsp-qubo}
H(\boldsymbol{x}) = H_{\text{obj}}(\boldsymbol{x}) + H_c(\boldsymbol{x}), \quad \boldsymbol{x} \in \{0,1\}^{N^2}.
\end{equation}
When $\lambda \gg \max_{i,j} d_{i,j}$, constraint satisfaction is prioritized, and the minimizing solution corresponds to a valid tour with approximately minimal length. This binary quadratic form can be directly mapped onto the Ising Hamiltonian implemented by a quantum annealer.

\subsection{Annealing Schedule Optimization}
On a D-Wave QPU, once a QUBO problem is mapped onto the hardware topology via minor embedding, the embedding is fixed. 
At this stage, performance improvements rely on optimizing the annealing schedule $s(t)$ under the given embedding. 
Following the optimal control formulation in~\cite{chen2022optimizing}, the schedule design is expressed as
\begin{equation}\label{eq:aso}
\arg\min_{s(t)} \ \langle \psi(T) \mid H \mid \psi(T) \rangle,
\end{equation}
where the goal is to minimize the expected energy of the final state. 

In general, $s(t)$ may take the form of a linear ramp or other parameterized families of schedules. 
In this work, we restrict to a Fourier-parameterized form, where the design vector includes both the annealing time $T$ and Fourier coefficients $\{\theta_m\}$. 
The optimization problem can thus be reformulated as
\begin{equation}\label{eq:aso-fourier}
\arg\min_{T,\,\{\theta_m\}} \ \langle \psi(T;\,\theta_1,\dots,\theta_M) \mid H \mid \psi(T;\,\theta_1,\dots,\theta_M) \rangle,
\end{equation}
where $M$ represents the number of Fourier coefficients used in the parameterization of the annealing schedule. This explicitly highlights that both $T$ and the Fourier parameters are optimized jointly in our approach, which sets the stage for the methodology introduced in the next section.


\section{Methodology}
\label{sec:met}
In this section, we present the proposed TuRBO framework for annealing schedule design on a D-Wave quantum annealer. 
Given a QUBO formulation of the TSP and its fixed minor embedding onto the hardware topology, we focus on optimizing the annealing schedule under that embedding. 
The schedule is parameterized in a Fourier basis as a control signal $s(t)$, and both the total annealing time $T$ and the Fourier coefficients $\boldsymbol{\theta}$ are optimized using TuRBO. Here, $\boldsymbol{\theta} = \{\theta_1, \theta_2, \ldots, \theta_M\}$ represents the set of Fourier coefficients that parameterize the annealing schedule.
At each iteration, TuRBO proposes a candidate schedule $(T,\boldsymbol{\theta})$, maps it to a device-valid signal through amplitude clipping and discretization, evaluates it on the QPU, and then updates a Gaussian process surrogate that guides the next proposal. 
Our methodology consists of three key elements: (i) a Fourier-parameterized annealing schedule, (ii) an objective and evaluation pipeline, and (iii) TuRBO. 
Figure~\ref{fig:Framework} provides an overview of the overall workflow.

\begin{figure*}[t]%
\begin{center}
\includegraphics[width=\linewidth]{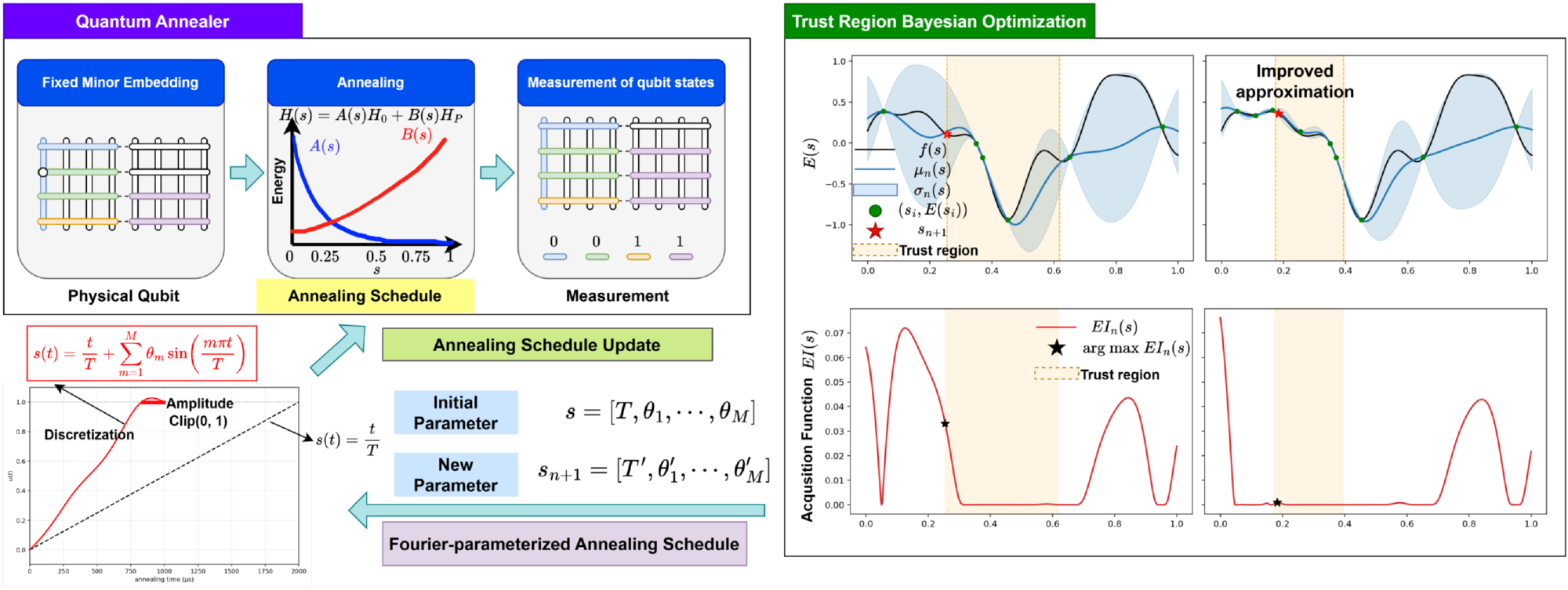}
\end{center}
\caption{TuRBO for annealing schedule on quantum annealer. 
(\textbf{Left Box}) QA workflow: fixed minor embedding, annealing with schedule $s(t)$, and measurement of qubit states. 
(\textbf{Right Box}) TuRBO: a Gaussian process surrogate models the objective energy $E(\boldsymbol{s})$, the acquisition function (Expected Improvement) selects new candidates within a trust region, and the surrogate is iteratively updated to refine the approximation and improve annealing schedules. 
The annealing schedule is parameterized in a Fourier basis. At each iteration, an initial parameter $\boldsymbol{s}$ is updated to a new parameter $\boldsymbol{s}'$ by TuRBO, which is then mapped to a device-valid schedule through amplitude clipping and discretization before being executed on the quantum annealer.
}
\label{fig:Framework}%
\end{figure*}

\subsection{Fourier-parameterized Annealing Schedule}
We parameterize the schedule as a truncated Fourier series:
\begin{equation}
s(t) = \frac{t}{T} + \sum_{m=1}^{M} \theta_m \sin\!\left(\frac{m \pi t}{T}\right),
\label{eq:fourier-schedule}
\end{equation}
where $T$ is the total annealing time, $M$ is the total number of Fourier components and $\{\theta_m\}$ are the Fourier coefficients. We define the TuRBO design vector as
\begin{equation}
\label{eq:design-vector}
\boldsymbol{s} \;=\; \big(T,\; \theta_1,\dots,\theta_M\big) \in \mathcal{S},
\end{equation} 
where $\mathcal{S}$ denotes the feasible space of the design vector, constrained by the problem and hardware limitations. We jointly optimize both $T$ and $\{\theta_m\}$ using TuRBO to improve solution quality under realistic hardware constraints.

To ensure a physically valid anneal, we apply the following post-processing before submission to the QPU:
(i) amplitude clipping $s(t)\!\leftarrow\!\min\{1,\max\{0,s(t)\}\}$ to enforce the valid range $[0,1]$,  
(ii) discretization of $s(t)$ onto the device control grid, taking into account the finite time and amplitude resolution supported by the hardware.

\subsubsection{Bounds}
We constrain the total annealing time $T$ within the device limits \( T \in [T_{\min}, T_{\max}] \). Additionally, we impose frequency-dependent box constraints on the Fourier coefficients:
\begin{equation}\label{eq:four}
\theta_m \in \Big[-\frac{\alpha}{m},\; \frac{\alpha}{m}\Big], \qquad m=1,\dots,M,
\end{equation}
where $\alpha$ is a constant, and the constraint serves as a low-pass regularizer, controlling the size of the search space for the Fourier coefficients.

\subsubsection{Discretization}
Since he function $s(t)$ cannot be passed to the QPU as a continuous curve. Instead, it must be specified on a finite grid of control points with limited time and amplitude resolution. 
In practice, we resample $s(t)$ onto this hardware-defined grid, rounding amplitudes to the nearest representable value and aligning time breakpoints to the supported discretization step. 
This ensures that the Fourier-parameterized schedule is faithfully mapped to a device-valid control signal.

\begin{remark}
The complexity of annealing schedule optimization grows with both the total annealing time $T$ and the number of Fourier coefficients $M$. 
A larger $M$ increases the expressiveness of the schedule but also expands the search space. 
Similarly, while longer annealing times $T$ may improve adiabaticity, hardware imposes practical limits. 
For instance, the D-Wave Advantage system typically allows annealing times between $1\ \mu\mathrm{s}$ and $2000\ \mu\mathrm{s}$, and the schedule $s(t)$ can only be specified at a finite number of control points with limited resolution~\cite{dwave}. 
Overshoots due to Fourier terms may also push $s(t)$ outside the valid $[0,1]$ range, requiring amplitude clipping before submission. 
Therefore, schedule optimization must balance expressive parametrization, hardware-imposed runtime bounds, and the precision of schedule control.
\end{remark}

\subsection{Objective and Evaluation Pipeline}
As defined in \eqref{eq:aso}, the goal of annealing schedule optimization is to minimize the expected energy of the final state. 
In practice, for each candidate schedule $\boldsymbol{s} = (T,\boldsymbol{\theta})$, the QPU is executed with a user-specified number of number of reads (shots), denoted $R$. 
This produces a sample set $\{\boldsymbol{x}^{(r)}\}_{r=1}^{R}$, where each $\boldsymbol{x}^{(r)} = (x_{1,1},\dots,x_{N,N}) \in \{0,1\}^{N^2}$ is a binary assignment of the decision variables defined in the TSP formulation, returned by the annealer.
The energy of each sample is evaluated by substituting $\boldsymbol{x}^{(r)}$ into the QUBO formulation \eqref{eq:tsp-qubo}. 
We then define the TuRBO objective as the minimum energy observed among the $R$ shots:
\begin{equation}
E(\boldsymbol{s}) \;=\; \min_{r=1,\dots,R} H\!\big(\boldsymbol{x}^{(r)};\boldsymbol{s}\big),
\label{eq:objective-energy}
\end{equation}
which serves as the direct feedback signal for the optimization loop. Here, $H(\boldsymbol{x}^{(r)};\boldsymbol{s})$ is the QUBO energy defined in \eqref{eq:tsp-qubo}. 
The TuRBO framework iteratively proposes new schedules $\boldsymbol{s}$, evaluates $E(\boldsymbol{s})$ on the QPU, and updates the surrogate model to guide subsequent exploration.

\begin{remark}
In practice, schedule optimization must also respect the QPU runtime limits. 
The total wall-clock time of an evaluation is determined not only by the annealing time $T$ in the schedule parameters $\boldsymbol{s}$, 
but also by the number of number of reads $R$ and hardware-specific overheads. 
On D-Wave systems, each call consists of a programming phase ($\sim$20 ms), followed by $R$ repetitions of anneal ($T$), readout ($\sim$100$\,\mu$s), and control overheads~\cite{dwave}. 
Thus the effective evaluation time can be approximated as
\[
t_{\text{eval}} \;\approx\; t_{\text{prog}} + R\cdot(t_{\text{anneal}}+t_{\text{readout}}+t_{\text{overhead}}).
\]
Consequently, when designing candidate schedules, both the annealing time $T$ and the readout budget $R$ must be chosen such that $t_{\text{eval}}$ does not exceed the QPU runtime limit. 
This constraint highlights the importance of balancing physical schedule parameters with practical execution time.
\end{remark}

\subsection{Trust Region Bayesian Optimization}
The standard BO is a black-box optimization framework designed for problems where the objective is expensive or noisy to evaluate. 
In our setting, the black-box function is the QPU-evaluated objective $E(\boldsymbol{s})$ defined in \eqref{eq:objective-energy}, 
which measures the minimum QUBO energy obtained from a set of annealing runs under a candidate schedule $\boldsymbol{s}=(T,\boldsymbol{\theta})$. 
Since each evaluation requires actual QPU queries, we cannot exhaustively explore the search space. 
Instead, BO builds a surrogate model that approximates the dependence of the energy on the schedule parameters $\boldsymbol{s}$, and uses this surrogate to decide where to query next.

However, evaluations on a quantum annealer are inherently noisy and expensive, and naive BO may waste queries by exploring unreliable regions. 
To mitigate this, we introduce a TuRBO strategy: the search is constrained to a local region around the incumbent best schedule, and this region is adaptively shrunk or expanded depending on recent progress. 
This makes BO more stable under stochastic hardware feedback and ensures efficient use of the limited QPU budget.

\subsubsection{Gaussian Process Surrogate}
The surrogate model explicitly captures the mapping between the schedule parameters \( \boldsymbol{s} \) and the objective function \( E(\boldsymbol{s}) \)
\begin{equation}
\boldsymbol{s} \mapsto E(\boldsymbol{s}),
\end{equation}
where $\boldsymbol{s}=(T,\theta_1,\dots,\theta_M)$ are the schedule parameters and $E(\boldsymbol{s})$ is the minimum QUBO energy obtained from samples under schedule $\boldsymbol{s}$. 
We employ a Gaussian process (GP) prior to model this mapping:
\begin{equation}
f \sim \mathcal{GP}(\mu_0, k(\cdot,\cdot)),
\end{equation}
where $\mu_0$ is the prior mean and $k(\cdot, \cdot)$ is the kernel function. The kernel function is defined as:
\begin{equation}
k(\boldsymbol{s},\boldsymbol{s}') = \sigma^2 \, k_{\text{Mat\'{e}rn 5/2}}(\boldsymbol{s},\boldsymbol{s}') + \sigma_n^2 \delta_{\boldsymbol{s}\boldsymbol{s}'},
\end{equation}
where $\boldsymbol{s}' = (T', \theta_1', \dots, \theta_M')$, $\sigma^2$ is the signal variance and $\sigma_n^2$ accounts for observation noise from stochastic QPU sampling. 
The Mat\'{e}rn 5/2 kernel is defined as
\begin{equation}
k_{\text{Mat\'{e}rn 5/2}}(\boldsymbol{s},\boldsymbol{s}') \;=\;
\Big(1+\sqrt{5}d+\tfrac{5}{3}d^2\Big)\exp(-\sqrt{5}d),
\end{equation}
where $d=\|\boldsymbol{s}-\boldsymbol{s}'\|/\ell$ is the scaled Euclidean distance and $\ell$ denotes the characteristic length-scale. We employ automatic relevance determination (ARD)~\cite{snoek2012practical, wang2023recent}, assigning an individual length-scale to each dimension of $\boldsymbol{s}$, 
which allows the GP to automatically learn which schedule parameters are more influential for the energy landscape.

\begin{remark}
We adopt the Mat\'{e}rn-$5/2$ kernel because it balances smoothness and flexibility. 
Unlike the squared exponential kernel that assumes infinitely differentiable functions, Mat\'{e}rn-$5/2$ models functions that are only twice differentiable, making it better suited for black-box objectives such as annealing schedule evaluations~\cite{finvzgar2024designing}. 
This is particularly relevant when optimizing Fourier-parameterized annealing schedules: 
small changes in the coefficients $\boldsymbol{\theta}$ or the total annealing time $T$ can induce non-linear and noisy variations in the observed energy. 
The Mat\'{e}rn-$5/2$ kernel has been widely adopted in BO for hyperparameter tuning and black-box optimization tasks~\cite{snoek2012practical}.
\end{remark}

Given observed data $\mathcal{D}_n=\{(\boldsymbol{s}_i, E(\boldsymbol{s}_i))\}_{i=1}^n$, where $n$ represents the number of observed schedule-energy pairs, the GP posterior provides, for any new candidate schedule $\boldsymbol{s}$, 
a predictive mean $\mu_n(\boldsymbol{s})$ and a predictive variance $v_n^2(\boldsymbol{s})$. 
The mean $\mu_n(\boldsymbol{s})$ serves as the surrogate estimate of the energy, while the variance $v_n^2(\boldsymbol{s})$ quantifies the model's uncertainty about this estimate. 
Intuitively, $\mu_n(\boldsymbol{s})$ guides exploitation by favoring schedules predicted to yield low energies, 
whereas $v_n^2(\boldsymbol{s})$ guides exploration by encouraging evaluation of schedules in regions where the surrogate is uncertain. 
Together, these two quantities form the basis for the acquisition function used to decide where to query next.

\subsubsection{Acquisition Function}
The predictive mean $\mu_n(\boldsymbol{s})$ and variance $v_n^2(\boldsymbol{s})$ from the GP posterior are combined through an acquisition function to decide which schedule parameters to evaluate next. 
We employ the expected improvement (EI) criterion~\cite{frazier2018tutorialbayesianoptimization}, which measures the expected amount by which a new evaluation at $\boldsymbol{s}$ will improve upon the incumbent best energy $f^*=\min_{i\le n} E(\boldsymbol{s}_i)$. 
This acquisition naturally balances exploitation and exploration: schedules with low predicted energy (small $\mu_n(\boldsymbol{s})$) are favored, but so are schedules with high predictive uncertainty (large $v_n^2(\boldsymbol{s})$).

Formally, the expected improvement is defined as:
\begin{equation}
\mathrm{EI}_n(\boldsymbol{s}) =
\big(f^* - \mu_n(\boldsymbol{s}) - \xi\big)\Phi\!\big(z(\boldsymbol{s})\big)
+ v_n(\boldsymbol{s})\phi\!\big(z(\boldsymbol{s})),
\end{equation}
where $z(\boldsymbol{s}) = \tfrac{f^* - \mu_n(\boldsymbol{s})-\xi}{v_n(\boldsymbol{s})}$, 
$\xi \ge 0$ is a hyperparameter that controls the trade-off between exploitation and exploration in the BO algorithm, and $(\Phi,\phi)$ are the Gaussian cumulative density function and probability density function. 
By convention, $\mathrm{EI}_n(\boldsymbol{s})=0$ if $v_n(\boldsymbol{s})=0$. 
The first term promotes exploitation, favoring candidates predicted to perform better than $f^*$, while the second term encourages exploration by focusing on regions where the model is uncertain. To determine the next query, we solve:
\begin{equation}
\boldsymbol{s}_{n+1} = \arg\max_{\boldsymbol{s}\in\mathcal{R}_n} \mathrm{EI}_n(\boldsymbol{s}),
\end{equation}
where $\mathcal{R}_n \subseteq \mathcal{S}$ denotes the current trust region centered at the incumbent best solution. 
Here, the feasible set $\mathcal{S}$ is defined by hardware-imposed box constraints, namely 
$T \in [T_{\min},T_{\max}]$ (typically $T_{\min}=1\,\mu\text{s}$ and $T_{\max}=2000\,\mu\text{s}$ on D-Wave quantum annealer), and frequency-dependent bounds on the Fourier coefficients as defined in~\eqref{eq:four}.
The maximization is performed using multi-start local search. 
Restricting the acquisition optimization to $\mathcal{R}_n$ ensures that candidate schedules satisfy device limits while focusing exploration around promising regions, and the adaptive expansion or contraction of the trust region preserves the ability to escape local optima under noisy hardware feedback. Finally, the bserved dataset is updated as follows:
\begin{equation}
    \mathcal{D}_{n+1}\gets \mathcal{D}_n \cup \{(\boldsymbol{s}_{n+1},E(\boldsymbol{s}_{n+1}))\},
\end{equation}
where $\mathcal{D}_{n+1}$ is the updated dataset that includes the new candidate schedule and its corresponding energy.
\subsubsection{Trust Region}
Because QPU evaluations are noisy and expensive, unrestricted maximization of the acquisition function over the full domain $\mathcal{S}$ may result in inefficient exploration or instability. 
To address this, TuRBO restricts the search to a local trust region $\mathcal{R}_n$ centered at the current incumbent $\boldsymbol{s}^*$. 
We define the trust region as an axis-aligned hyper-rectangle in the schedule parameter space:
\begin{equation}
\mathcal{R}_n = \{\boldsymbol{s} : |s_j - s_j^*| \leq \Delta_j, \; j=1,\dots,M+1 \},
\end{equation}
where each \( s_j \) represents the \( j \)-th parameter of the annealing schedule, \( M+1 \) is the total number of schedule parameters, and \( \Delta_j \) denotes the side length along the \( j \)-th dimension. The side lengths \( \Delta_j \) are updated adaptively: if no improvement is observed over \( k \) iterations, \( \Delta_j \) is reduced by a factor \( \rho < 1 \), shrinking the region to encourage local refinement. Conversely, when improvement is observed, the region is expanded by \( \rho^{-1} \), up to the hardware-imposed box constraints.

This simple hyper-rectangle definition follows the principle of TuRBO methods such as TuRBO~\cite{eriksson2019scalable, namura2025regional}, which have shown that maintaining local search regions improves performance in noisy or high-dimensional settings. 
In our case, the number of schedule parameters $(T,\theta_1,\dots,\theta_M)$ is moderate, but the QPU noise makes adaptive restriction essential for stability. 
By concentrating evaluations within $\mathcal{R}_n$, we exploit local structure around promising schedules, while the shrink-expand mechanism preserves the ability to explore globally and escape local optima.
The detailed procedure is presented in Algorithm~\ref{alg:trbo-fourier}.

\subsection{Complexity Analysis}
In this subsection, e analyze the computational and memory complexity of the proposed TuRBO method for optimizing Fourier-parameterized annealing schedules.

\begin{theorem}[Computational Complexity of TuRBO (single-point) for Fourier-based Schedule Optimization]
\label{thm:turbo-complexity}
Let $\boldsymbol{s}=(T,\theta_1,\dots,\theta_M)$ be the Fourier-parameterized annealing schedule 
defined in~\eqref{eq:fourier-schedule}-\eqref{eq:design-vector}, and let the TuRBO objective be 
$E(\boldsymbol{s})$ as in~\eqref{eq:objective-energy}. 
Assume TuRBO proposes \emph{one} candidate per iteration (no batch), maintains $M_{\mathrm{TR}}$ trust regions, 
each capped at $n_{\max}$ observations, and maximizes the acquisition with $R_{\mathrm{acq}}$ random restarts in 
dimension $p=M{+}1$.

Then, the time complexity for a single iteration is:
\begin{equation}
T_{\mathrm{iter}}
= \mathcal{O}\!\Big(M_{\mathrm{TR}}\,n_{\max}^3 
   + R_{\mathrm{acq}}\,M_{\mathrm{TR}}\,C_{\mathrm{acq}}(d)\Big),
\end{equation}
and the memory complexity is:
\begin{equation}
S_{\mathrm{iter}}
= \mathcal{O}\!\big(M_{\mathrm{TR}}\,n_{\max}^2\big),
\end{equation}
where \( M_{\mathrm{TR}} \) is the number of trust regions, \( n_{\max} \) is the maximum number of observations per trust region, \( R_{\mathrm{acq}} \) is the number of random restarts used for acquisition maximization, \( C_{\mathrm{acq}}(p) \) represents the cost of one local acquisition maximization in a \( p \)-dimensional space, and the term \( n_j^2 \) captures the complexity of each GP posterior evaluation, including the computation of the mean, variance, and gradient.

Moreover, under standard smoothness and surrogate-accuracy assumptions, TuRBO attains an 
$\varepsilon$-stationary point in
\begin{equation}
\mathcal{O}(\varepsilon^{-2}).
\end{equation}
QPU evaluations (up to additional logarithmic factors in $1/\varepsilon$).
\end{theorem}

\noindent\textit{Proof.} See Appendix~\ref{appendix:turbo-proof}.

\begin{algorithm}[t]
\caption{TuRBO for Fourier-Parameterized Annealing Schedules}
\label{alg:trbo-fourier}
\begin{algorithmic}[1]
\Require QUBO $H$, fixed embedding, Fourier order $M$, feasible set $\mathcal{S}$ (box constraints), initial TR side lengths $\boldsymbol{\Delta}_{\text{init}}$, shrink/expand factor $\rho$, patience $k$, runtime budget $B$.
\Ensure Optimized schedule $\boldsymbol{s}^*=(T^*,\boldsymbol{\theta}^*)$

\State Initialize $n$ space-filling points $\{\boldsymbol{s}_i\}_{i=1}^{n}\subset\mathcal{S}$.
\For{$i=1$ to $n$}
    \State Construct $s(t)$ via \eqref{eq:fourier-schedule}, clip to $[0,1]$, discretize to QPU grid.
    \State Run QPU with $R_0$ number of reads under budget $B$, obtain $\{\boldsymbol{x}^{(r)}\}$.
    \State Compute $E(\boldsymbol{s}_i) = \min_r H(\boldsymbol{x}^{(r)};\boldsymbol{s}_i)$.
\EndFor
\State $\mathcal{D}_{n} \gets \{(\boldsymbol{s}_i,E(\boldsymbol{s}_i))\}$.
\State Set incumbent $\boldsymbol{s}^* \gets \arg\min E(\boldsymbol{s}_i)$, $f^* \gets E(\boldsymbol{s}^*)$.
\State Initialize trust region $\mathcal{R}_n$ centered at $\boldsymbol{s}^*$ with side lengths $\boldsymbol{\Delta}_{\text{init}}$.
\State no\_improve $\gets 0$.

\While{budget not exhausted}
    \State Fit GP with Matérn-5/2 (ARD) + white noise on $\mathcal{D}_n$.
    \State Select next schedule
        \[
        \boldsymbol{s}_{n+1} \gets \arg\max_{\boldsymbol{s}\in \mathcal{R}_n} \mathrm{EI}_n(\boldsymbol{s})
        \]
        using multi-start local search.
    \State Choose number of reads $R$ adaptively (small early, large near $\boldsymbol{s}^*$) under budget $B$.
    \State Build $s(t)$, clip, discretize, run QPU with $R$ number of reads.
    \State Compute $E(\boldsymbol{s}_{n+1})=\min_r H(\boldsymbol{x}^{(r)};\boldsymbol{s}_{n+1})$.
    \State Update $\mathcal{D}_{n+1}\gets \mathcal{D}_n \cup \{(\boldsymbol{s}_{n+1},E(\boldsymbol{s}_{n+1}))\}$.
    \If{$E(\boldsymbol{s}_{n+1}) < f^*$}
        \State $\boldsymbol{s}^* \gets \boldsymbol{s}_{n+1}$, $f^* \gets E(\boldsymbol{s}_{n+1})$, no\_improve $\gets 0$.
        \State Expand $\mathcal{R}_{n+1}$ by factor $\rho^{-1}$ within $\mathcal{S}$, re-center at $\boldsymbol{s}^*$.
    \Else
        \State no\_improve $\gets$ no\_improve $+1$.
        \If{no\_improve $\ge k$}
            \State Shrink $\mathcal{R}_{n+1}$ by factor $\rho$, re-center at $\boldsymbol{s}^*$.
            \State no\_improve $\gets 0$.
        \EndIf
    \EndIf
    \If{any side length of $\mathcal{R}_{n+1}$ $< \Delta_{\min}$}
        \State Restart: reinitialize trust region around incumbent or a new seed.
    \EndIf
\EndWhile
\State \Return $\boldsymbol{s}^*, E(\boldsymbol{s}^*)$
\end{algorithmic}
\end{algorithm}

\section{Results}
\label{sec:Exp}
This section presents an experimental evaluation of the proposed TuRBO of annealing schedules on a quantum annealer, including the setup, benchmarks, and performance across multiple TSP instances generated. All implementations, including reimplemented versions of some baselines, are made publicly available on our GitHub repository https://github.com/JeongQC/QASchedule.
\subsection{Experimental Settings}

\subsubsection{Data settings} 
We generate TSP instances by randomly placing cities in a two-dimensional Euclidean plane. 
The distance matrix is computed as the symmetric Euclidean distance between all pairs of cities. 
To ensure reproducibility, all instances are generated with a fixed random seed of 0, and we gradually increase the number of cities $N$ to evaluate scalability.

\subsubsection{Baselines} 
We compare the proposed TuRBO framework against the following baselines:

\begin{itemize}
    \item Schedule optimization baselines:
    \begin{enumerate}[]
        \item RS: At each iteration, a candidate schedule $\boldsymbol{s}=(T,\theta_1,\ldots,\theta_M)$ is drawn by independently sampling $T \sim \mathrm{LogUniform}[T_{\min},T_{\max}]$ and $\theta_m \sim \mathrm{Uniform}[-1/m,,1/m]$. Each candidate is evaluated on the QPU under a fixed read-budget profile, and the incumbent with the lowest tour energy is retained.
        \item GS: The search is initialized at the geometric mean annealing time $T_{\text{geo}}$ and zero Fourier coefficients ($\theta_m{=}0$). At each step, one coordinate from $\{T,\theta_1,\ldots,\theta_M\}$ is randomly selected, and $\pm$ perturbations are tested (log-domain steps for $T$, clipped linear steps for $\theta_m$). A move is accepted only if the objective energy decreases; otherwise, the step size is gradually reduced until the evaluation budget is exhausted.
    \end{enumerate}
    \item General optimization baselines:
    \begin{enumerate}
        \item SA: This algorithm is implemented using the default configuration from the D-Wave Ocean library, with the objective of minimizing the total tour length encoded in the TSP QUBO energy.
        \item GA: A standard evolutionary search algorithm is applied to the Quadratic QUBO formulation of the TSP, with the objective function defined as the minimization of the total tour distance.
        \item CPLEX: A state-of-the-art classical solver for mixed integer programming (MIP), utilized to provide strong baseline solutions for comparison.
    \end{enumerate}
\end{itemize}

\subsubsection{Implementation details}
All algorithms are implemented in Python 3.8 using the official D-Wave Ocean SDK and standard optimization libraries. 
Schedule optimization methods (TuRBO, RS, GS) are executed on a local workstation equipped with an AMD Ryzen 9 7950X processor, 127GB of RAM, and running Windows 11. 
Classical optimization baselines, including GA, SA, and CPLEX, are also executed on the same local machine; SA is run with the 2000 num\_read and default configuration in the Ocean library, and CPLEX is executed with its default MIP solver settings. 
QA evaluations of candidate schedules is executed on the \textit{D-Wave Advantage2\_System1.6} (zephyr, 4800 qubits) via the D-Wave Leap cloud platform. 
The annealing time $T$ and the number of reads $R$ are selected within the hardware limits of the QPU, and CQM is run with the default settings provided in the Ocean library. 
For fairness, all baselines and the proposed method are evaluated on the same TSP instances generated with random seed 0, and all reported metrics correspond to the averages over 10 independent runs.

\subsection{Evaluation Metrics}

We use the following evaluation metrics to assess the performance of the proposed framework and the baseline methods:

\begin{itemize}
    \item Energy ($E$): The minimum tour length, also referred to as the QUBO energy, as defined in equation \eqref{eq:tsp-qubo}.
    \item Success Probability ($p_{\mathrm{succ}}$): The fraction of valid solutions that satisfy the constraints of the TSP out of all solutions generated by the algorithm.
    \item CBF: The average fraction of broken qubit chains in the minor embedding, as reported by the D-Wave sampler.
     \item Time to Solution ($TTS$): The expected time to observe at least one ground state solution with confidence $Q$ (default $Q{=}0.99$). 
    Let $T_f$ denote the annealing time per run and $P(T_f)$ the per-run ground state success probability, i.e., the fraction of runs in which at least one read attains the ground-state energy $E^*$. 
    The probability of at least one ground-state success after $R$ runs (readouts) is $Q = 1 - \bigl(1 - P(T_f)\bigr)^R$
    and the corresponding $TTS$ is expressed as:
    \[
      TTS \;=\; T_f \frac{\ln(1-Q)}{\ln\!\bigl(1 - P(T_f)\bigr)}.
    \]
    Note that $P(T_f)$ is computed from the per-read ground-state probability; this differs from the per-read feasible success $p_{\mathrm{succ}}$ reported above. 
\end{itemize}

\subsection{Research Questions}

We design experiments to answer the following research questions (RQs):

\begin{description}
    \item[RQ1:] What is the maximum size $N$ (number of cities) of the TSP that can be directly embedded on the D-Wave Advantage2 QPU without additional decomposition or hybrid algorithms? 
    This question evaluates the intrinsic hardware embedding capacity.

    \item[RQ2:] How effective is the proposed TuRBO compared to baselines (RS, GS)? 
    We compare in terms of minimum energy, success probability, CBF, and TTS.

    \item[RQ3:] Does reducing the control/search space by fixing the annealing time \( T \) and the readout counts lead to more effective schedule optimization? We optimize only the schedule shape (Fourier coefficients) and compare the following two approaches: (i) joint optimization of both \( T \) and \( \boldsymbol{\theta} \), and (ii) Fourier-only schedule optimization, where \( T \) and the readout counts are fixed and only \( \boldsymbol{\theta} \) is optimized. We evaluate the effectiveness of each approach using the following metrics: best energy, \( p_{\text{succ}} \), CBF, and reconstructed annealing schedules.

    \item[RQ4:] When applying the proposed TuRBO schedule optimization, can QA find high-quality solutions faster than classical algorithms? 
    We evaluate computation time, contrasting QPU-based results with classical baselines.
\end{description}


\subsection{RQ1: Embedding Scalability}
\begin{table}[h]
\centering
\caption{Zephyr clique-embedding statistics for TSP with $N$ customers.}
\label{tab:clique-embed}
\begin{tabular}{c c c c}
\hline
$N$ & Logical variables & Physical qubits & Chain length (mean) \\
\hline
2 & 4   & 4    & 1.00 \\
3 & 9   & 22   & 2.44 \\
4 & 16  & 40   & 2.50 \\
5 & 25  & 88   & 3.52 \\
6 & 36  & 161  & 4.47 \\
7 & 49  & 267  & 5.45 \\
8 & 64  & 416  & 6.50 \\
9 & 81  & 687  & 8.48 \\
10 & 100 & 1151 & 11.51 \\
\hline
\end{tabular}
\end{table}
We evaluate embedding scalability for the TSP QUBO on the \textit{D-Wave Advantage2\_System1.6} (Zephyr topology, 4800 qubits) using the \textit{DWaveCliqueSampler}. Because practical QPUs, including Zephyr, use sparse, bounded-degree hardware graphs, we emulate all-to-all connectivity by embedding a clique (complete graph), which provides a clean, device-level upper-bound baseline independent of problem sparsity~\cite{pelofske2025comparing}. For an instance with \(N\) cities, the assignment formulation has \(N^2\) binary variables; for each \(N\), we embed the complete graph \(K_{N^2}\) on the Zephyr topology and record three quantities from the resulting embedding: the logical-variable count \(N^2\), the number of unique physical qubits used, and the mean chain length, where a chain denotes the connected set of physical qubits assigned to a single logical variable. We query the sampler for its \textit{largest\_clique\_size} \(K_{\max}\) and set the test range accordingly: the clique required by size \(N\) fits only if \(N^2 \le K_{\max}\) (equivalently, \(N_{\max}=\lfloor\sqrt{K_{\max}}\rfloor\)). In our setup this yields \(N_{\max}=10\); therefore all subsequent experiments use \(N \le 10\). In our measurements, the computation time for clique embedding on Zephyr topology is approximately $0.01$\,s for all sizes, while the physical-qubit footprint and the mean chain length grow sharply with \(N\) (Table~\ref{tab:clique-embed}); for example, physical qubits increase from $4$ at \(N{=}2\) to $1151$ at \(N{=}10\) and the mean chain length from $1.0$ to $11.5$, indicating increasing sensitivity to noise and chain breaks~\cite{venturelli2015quantum,grant2022benchmarking}. A detailed theoretical analysis of these embedding scalings, including separator-based lower bounds, constructive upper bounds, 
and their implications for Table~\ref{tab:clique-embed}, is provided in Appendix~\ref{app:chain-scaling}.

\subsection{RQ2: Effectiveness of Trust-Region Bayesian Optimization}
\label{sec:rq2}
We evaluate the effectiveness of TuRBO for annealing-schedule design of the TSP QUBO on \textit{D-Wave Advantage2\_System1.6} under a fixed minor embedding and a Fourier parameterization \(s(t)=t/T+\sum_{m=1}^{8}\theta_m\sin(m\pi t/T)\). For budget fairness, TuRBO is compared with RS and GS using the same number of candidate schedules and an identical number of reads schedule (from $250$ to $900$), while feasibility and energy are evaluated identically across methods. We report feasible best energy, $p_{\text{succ}}$, CBF at best energy, annealing time ($T_{us}$), and readout counts ($\# ~\text{readout}$). 

\begin{table}[t]
\centering
\caption{Summary of TSP results across methods (TuRBO, RS, GS).}
\label{tab:tsp_results}
\begin{tabular}{c l c c c c c}
\hline
$N$ & Method & Best Energy & $p_{\text{succ}}$ & CBF & $T_{us}$ & $\# ~\text{readout}$ \\
\hline
  & TuRBO & 163.68 & 1.0000 & 0.0000 & 20.00 & 250 \\
2 & RS    & 163.68 & 1.0000 & 0.0000 & 21.73 & 250 \\
  & GS    & 163.68 & 1.0000 & 0.0000 & 20.00 & 250 \\
\hline
  & TuRBO & 206.91 & 1.0000 & 0.0000 & 20.00 & 250 \\
3 & RS    & 206.91 & 0.9170 & 0.0000 & 30.96 & 250 \\
  & GS    & 206.91 & 0.9520 & 0.0000 & 20.00 & 250 \\
\hline
  & TuRBO & 238.06 & 0.7600 & 0.0000 & 20.00 & 250 \\
4 & RS    & 238.06 & 0.6153 & 0.0000 & 34.46 & 250 \\
  & GS    & 238.06 & 0.6230 & 0.0000 & 20.00 & 250 \\
\hline
  & TuRBO & 249.43 & 0.6891 & 0.0000 & 20.00 & 250 \\
5 & RS    & 249.43 & 0.5875 & 0.0000 & 22.74 & 250 \\
  & GS    & 249.43 & 0.5903 & 0.0000 & 20.94 & 250 \\
\hline
  & TuRBO & 271.75 & 0.3441 & 0.0000 & 20.00 & 250 \\
6 & RS    & 271.75 & 0.2678 & 0.0000 & 332.49 & 283 \\
  & GS    & 271.75 & 0.3272 & 0.0000 & 194.94 & 250 \\
\hline
  & TuRBO & 327.27 & 0.0441 & 0.0005 & 56.29 & 250 \\
7 & RS    & 327.27 & 0.0276 & 0.0153 & 183.44 & 269 \\
  & GS    & 327.27 & 0.0408 & 0.0151 & 194.94 & 250 \\
\hline
  & TuRBO & 265.24 & 0.0098 & 0.0003 & 21.62  & 306 \\
8 & RS    & 265.24 & 0.0058 & 0.0004 & 48.53  & 343 \\
  & GS    & 265.24 & 0.0058 & 0.0152 & 194.94 & 343 \\
\hline
  & TuRBO & 465.26 & 0.0003 & 0.0145 & 56.29  & 306 \\
9 & RS    & --     & --   & --     & --     & --   \\
  & GS    & --     & --   & --     & --     & --   \\
\hline
  & TuRBO & --     & --   & --     & --     & --   \\
10 & RS    & --     & --   & --     & --     & --   \\
  & GS    & --     & --   & --     & --     & --   \\
\hline
\end{tabular}
\end{table}

As shown in Table~\ref{tab:tsp_results}, across all tested sizes up to \(N{=}8\), TuRBO, RS, and GS successfully found optimal solutions. 
The annealing times differ because GS, once it reaches an optimum from its initial random seed, terminates exploration and does not search for better schedules, whereas RS and BO continue to evaluate additional candidates. 
This continued exploration alters the annealing schedule, which can increase annealing time but also improves success probability. 
Up to \(N{=}6\), all methods maintain successful chain connectivity, but from \(N{=}7\) onward, CBFs become more significant, strongly influencing both the annealing process and the final solution quality. 
At \(N{=}7\), TuRBO achieves a lower CBF than GS, indicating that it can mitigate chain breaks and improve stability at larger problem sizes. 
From $N{=}9$, however, none of the methods were able to find an optimal solution, and at \(N{=}10\) even feasible solutions could not be obtained. 
We defer a detailed discussion of these scalability limitations to Section~\ref{sec:discussion}.  

These results indicate that beyond this point, careful schedule design becomes increasingly important, with TuRBO showing clear advantages by reducing CBF, improving stability, and achieving better overall performance.

\subsection{RQ3: Impact of Reducing the Control Space for Schedule Optimization}
\begin{table}[t]
\centering
\caption{Summary of TSP results across methods (TuRBO, RS, GS) with fixed annealing time ($T_{us}{=}250$) and readout counts ($2000$).}
\label{tab:tsp_results2}
\begin{tabular}{c l c c c c c}
\hline
$N$ & Method & Best Energy & $p_{\text{succ}}$ & CBF & $T_{us}$ & $\# ~\text{readout}$ \\
\hline
  & TuRBO & 163.68 & 1.0000 & 0.0000 & 250 & 2000 \\
2 & RS & 163.68 & 1.0000 & 0.0000 & 250 & 2000 \\
  & GS & 163.68 & 1.0000 & 0.0000 & 250 & 2000 \\
\hline
  & TuRBO & 206.91 & 1.0000 & 0.0000 & 250 & 2000 \\
3 & RS & 206.91 & 1.0000 & 0.0000 & 250 & 2000 \\
  & GS & 206.91 & 1.0000 & 0.0000 & 250 & 2000 \\
\hline
  & TuRBO & 238.06 & 0.8518 & 0.0000 & 250 & 2000 \\
4 & RS & 238.06 & 0.8889 & 0.0000 & 250 & 2000 \\
  & GS & 238.06 & 0.8275 & 0.0000 & 250 & 2000 \\
\hline
  & TuRBO & 249.43 & 0.6211 & 0.0041 & 250 & 2000 \\
5 & RS & 249.43 & 0.5252 & 0.0048 & 250 & 2000 \\
  & GS & 249.43 & 0.5816 & 0.0041 & 250 & 2000 \\
\hline
  & TuRBO & 271.75 & 0.2087 & 0.0060 & 250 & 2000 \\
6 & RS & 271.75 & 0.1274 & 0.0098 & 250 & 2000 \\
  & GS & 271.75 & 0.1848 & 0.0063 & 250 & 2000 \\
\hline
  & TuRBO & 327.27 & 0.0446 & 0.0056 & 250 & 2000 \\
7 & RS & 327.27 & 0.0305 & 0.0066 & 250 & 2000 \\
  & GS & 327.27 & 0.0293 & 0.0057 & 250 & 2000 \\
\hline
  & TuRBO & 265.40 & 0.0052 & 0.0058 & 250 & 2000 \\
8 & RS & 265.60 & 0.0031 & 0.0071 & 250 & 2000 \\
  & GS & 273.80 & 0.0041 & 0.0060 & 250 & 2000 \\
\hline
  & TuRBO & 446.98 & 0.0005 & 0.0051 & 250 & 2000 \\
9 & RS & 525.34 & 0.0005 & 0.0056 & 250 & 2000 \\
  & GS & --     & --      & --     & --  & --   \\
\hline
  & TuRBO & 529.65 & 0.0005 & 0.0097 & 250 & 2000 \\
10 & RS & --     & --      & --     & --  & --   \\
  & GS & --     & --      & --     & --  & --   \\
\hline
\end{tabular}
\end{table}

We investigate whether reducing the control space by fixing the annealing time $T$ and the readout counts enables more effective schedule optimization. 
In this setting, TuRBO optimizes only the schedule shape, parameterized by the Fourier coefficients $\boldsymbol{\theta}$, while $T$ and the number of number of reads are held constant. 
All methods are evaluated under identical QPU budgets, and we measure best energy, success probability $p_{\text{succ}}$, and CBF across repetitions.

\begin{figure}[!t]
\centering
\subfigure[Schedules at $N{=}7$.]{{\includegraphics[width=0.45\textwidth]{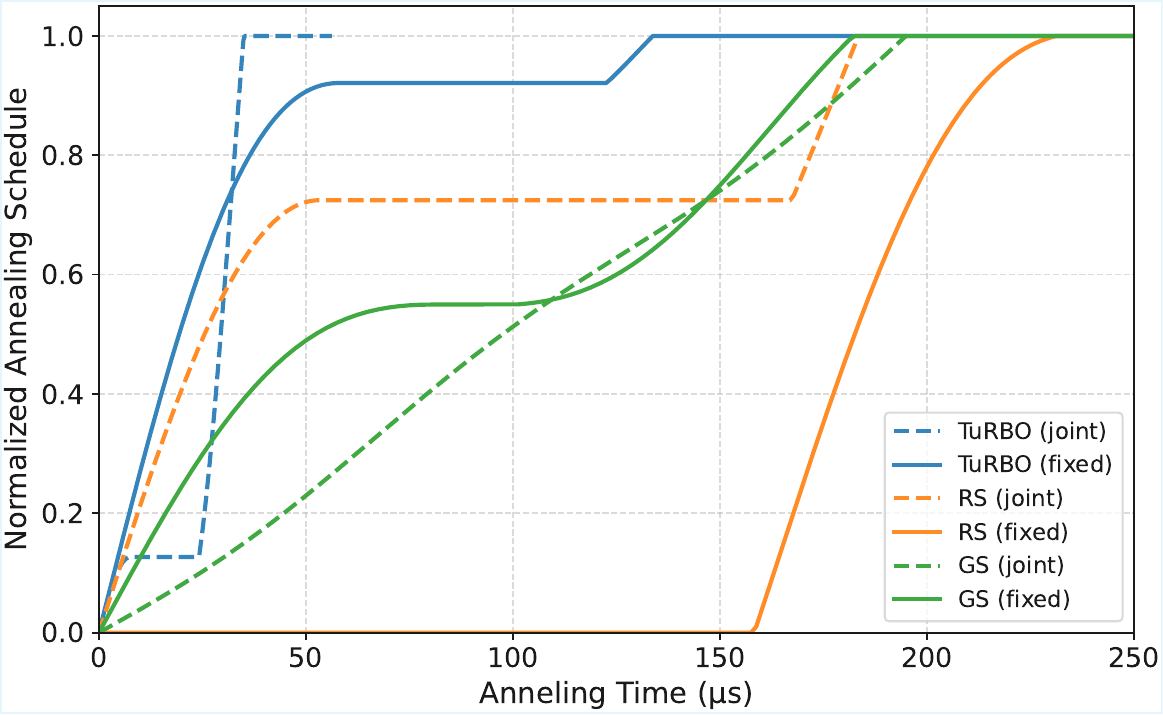} }}%
\label{fig.N6}
\vspace{0cm}
\subfigure[Schedules at $N{=}8$.]{{\includegraphics[width=0.45\textwidth]{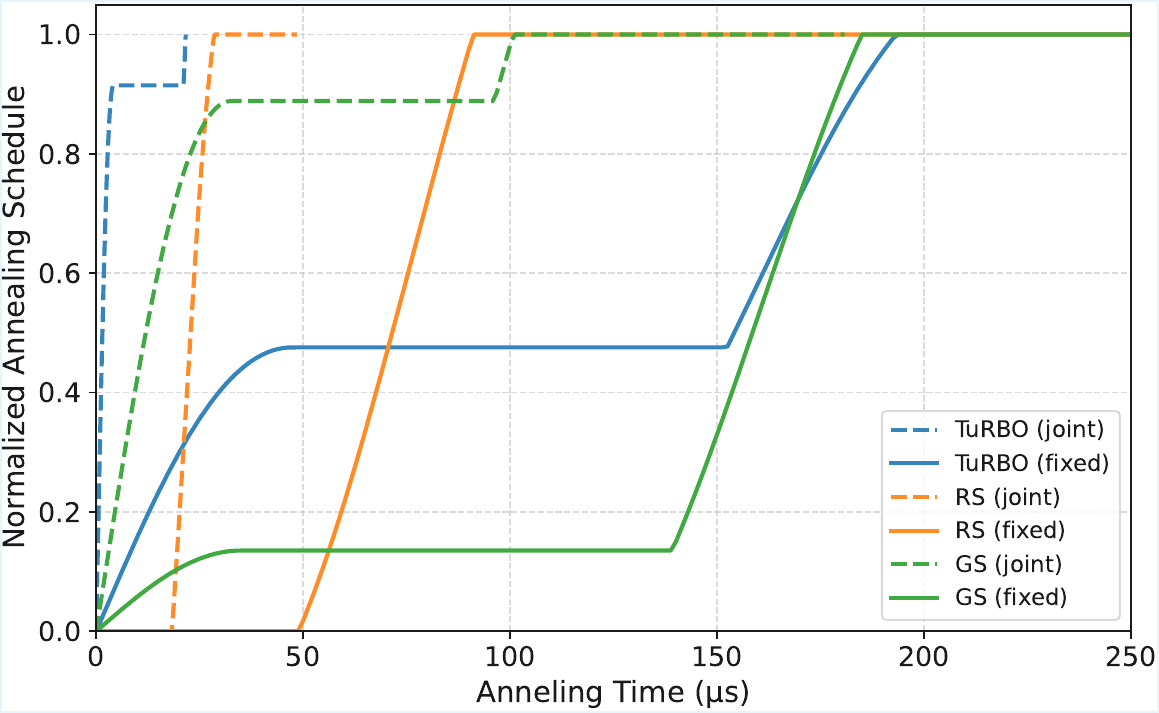} }}%
\label{fig.N7}
\caption{Comparison of normalized annealing schedules for $N{=}7$ and $N{=}8$ under fixed vs.\ joint optimization across TuRBO, RS, and GS. Line styles distinguish fixed (solid) vs.\ joint (dashed) control, and colors correspond to optimization methods.}%
\label{fig:schedules_n6n7}%
\end{figure}

As shown in Table~\ref{tab:tsp_results2}, fixing $T$ and number of reads allows all methods to solve instances up to $N{=}10$, which is larger than the cases addressed under joint optimization in Table~\ref{tab:tsp_results}. 
However, while Table~\ref{tab:tsp_results} shows that all methods successfully reached the true optimum at $N{=}8$ under joint optimization, none of the methods achieved the optimum at $N{=}7$ when $T$ and number of reads were fixed. 
From $N{=}9$ onward, only sub-optimal or infeasible solutions were obtained. 
This highlights a clear trade-off: reducing the search space enables exploration of larger problem sizes, but at the cost of losing optimality.

As seen in Tables~\ref{tab:tsp_results} and \ref{tab:tsp_results2}, for $N{=}5$ and $6$, the joint optimization setting exhibits essentially no chain breaks, whereas the chain break occurs in the fixed control setting at a nearly similar level that remains roughly flat as $N$ increases. This pattern suggests that when the annealing time is held constant, schedule-shape adjustments alone have limited leverage over chain connectivity. By contrast, under joint optimization at $N{=}7$ and $8$ (Figure~\ref{fig:schedules_n6n7}), allowing both the annealing time and the schedule shape to vary induces clear method-dependent differences in CBF, indicating that tuning $T$ in tandem with the shape is a key driver for mitigating chain breaks under fixed embedding. 
In particular, Figure~\ref{fig:schedules_n6n7} shows that the joint optimization tends to select shorter annealing times with steeper schedule shapes compared to the fixed setting. 
Such sharper profiles, while enabling faster evolution, appear to correlate with increased CBF, suggesting that the interplay between reduced anneal time and aggressive schedule shaping directly influences chain connectivity under the fixed embedding. This observation is similar to the findings reported in~\cite{jeong2025embeddingawarenoisemodelingquantum}.

These findings indicate that reducing the control space provides a pragmatic path to scale up problem size, but it does not guarantee optimality; joint optimization of both $T$ and schedule shape remains important for suppressing chain breaks and attaining true optima under larger instances.

\subsection{RQ4: Runtime and Optimality Gap Comparison}
\begin{table}[t]
\centering
\caption{Performance across methods with gap to CPLEX optimal (\%). Computation time in ms.}
\label{tab:fair_comp}
\begin{tabular}{c l c c c}
\hline
$N$ & Method & Gap (\%) & Best energy & Time (ms) \\
\hline
2 & SA     & 0.00 & 163.68 &  52.28 \\
  & GA     & 0.00 & 163.68 & 2408.36 \\
  & CPLEX  & 0.00 & 163.68 &  15.26 \\
  & QA (TuRBO)     & 0.00 & 163.68 &   5.00 \\
\hline
3 & SA     & 0.00 & 206.91 & 203.11 \\
  & GA     & 0.00 & 206.91 & 2480.26 \\
  & CPLEX  & 0.00 & 206.91 &  12.04 \\
  & QA (TuRBO)     & 0.00 & 206.91 &   5.00 \\
\hline
4 & SA     & 0.00 & 238.06 & 308.64 \\
  & GA     & 0.00 & 238.06 & 2539.41 \\
  & CPLEX  & 0.00 & 238.06 &  32.95 \\
  & QA (TuRBO)     & 0.00 & 238.06 &   5.00 \\
\hline
5 & SA     & 0.00 & 249.43 & 469.79 \\
  & GA     & 0.00 & 249.43 & 2568.70 \\
  & CPLEX  & 0.00 & 249.43 &  47.22 \\
  & QA (TuRBO)     & 0.00 & 249.43 &   5.00 \\
\hline
6 & SA     & 14.99 & 312.49 & 702.76 \\
  & GA     &  0.00 & 271.75 & 2620.02 \\
  & CPLEX  &  0.00 & 271.75 & 121.40 \\
  & QA (TuRBO)     &  0.00 & 271.75 &   5.00 \\
\hline
7 & SA     & 0.69 & 329.52 &  844.15 \\
  & GA     & 0.00 & 327.27 & 2665.16 \\
  & CPLEX  & 0.00 & 327.27 &  175.05 \\
  & QA (TuRBO)     & 0.00 & 327.27 &    14.07 \\
\hline
8 & SA     & 0.06 & 265.40 & 1325.17 \\
  & GA     & 0.00 & 265.24 & 2696.38 \\
  & CPLEX  & 0.00 & 265.24 &  111.87 \\
  & QA (TuRBO)     & 0.00 & 265.24 &    6.62 \\
\hline
9 & SA     & 0.86 & 308.46 & 1510.34 \\
  & GA     & 0.00 & 305.84 & 2742.69 \\
  & CPLEX  & 0.00 & 305.84 &  196.09 \\
  & QA (TuRBO)     & 46.15 & 446.98 &  500.00 \\
\hline
10 & SA     & 12.34 & 385.15 & 2265.11 \\
  & GA     &  0.00 & 342.85 & 2773.54 \\
  & CPLEX  &  0.00 & 342.85 &  544.94 \\
  & QA (TuRBO)     & 54.48 & 529.65 &  500.00 \\
\hline
\end{tabular}
\end{table}

We evaluate SA, GA, CPLEX, and the proposed QA (TuRBO) on generated TSP instances. The optimality gap is defined as
\[
\mathrm{Gap}(\%) \;=\; 100\times\frac{E_{\text{method}} - E_{\text{CPLEX}}}{E_{\text{CPLEX}}}\!,
\]
where $E_{\text{method}}$ denotes the best objective obtained by each method 
and $E_{\text{CPLEX}}$ denotes the optimal objective found by CPLEX. For QA (TuRBO), the reported run time corresponds to the defined TTS metric.

Table~\ref{tab:fair_comp} reports the best energy, the gap with respect to CPLEX, and the run time (ms). Across small-medium sizes ($N{=}2$-$5$), all methods reach the CPLEX optimum ($0$\% gap). 
The QA (TuRBO) consistently solves these instances with nearly constant annealing times ($5$\,ms), in contrast to SA (tens to hundreds of milliseconds), GA (several seconds), and CPLEX (tens of milliseconds).

For $N{=}6$-$8$, QA (TuRBO) remains optimal (0\% gap); by contrast, SA shows a $14.99$\% gap at $N{=}6$ and small nonzero gaps at $N{=}7$ and $8$. GA and CPLEX remain optimal in this range. Remarkably, QA (TuRBO) maintains short and stable runtimes ($5.00$-$14.07$\,ms), whereas SA increases to the order of seconds, GA remains around $2.6$-$2.7$\,s, and CPLEX requires $111.87$-$175.05$\,ms.

For larger instances ($N{=}9$ and $10$), QA (TuRBO) exhibits sizable optimality gaps ($46.15$\% and $54.48$\%), while GA and CPLEX remain optimal, and SA is closer to optimal than QA (TuRBO). This degradation aligns with hardware embedding limits and the difficulty of maintaining chain connectivity and adiabaticity as problem size grows. Notably, QA (TuRBO) still executes quickly ($500$\,ms) but yields suboptimal energies.

These results indicate that within the embedding-feasible range, properly tuned annealing schedules enable proposed QA to solve instances with nearly constant, polynomial-scale effort and superior runtime compared to classical algorithms. However, on current hardware, once clique embeddings approach their maximum capacity, proposed QA fails to recover the optimal solution. We further elaborate on these limitations in Section~\ref{sec:discussion}.

\section{Discussion}\label{sec:discussion}

In Section~\ref{sec:Exp}, the failure to obtain optimal solutions from \(N{=}9\) and the absence of feasible solutions at \(N{=}10\) highlight the scalability limitations of the current approach on physical quantum annealers. 
We attribute these difficulties primarily to two factors: (i) the embedding overhead required to map the fully connected TSP QUBO onto the sparse Zephyr topology, which increases chain lengths and susceptibility to errors, and (ii) noise effects intrinsic to the hardware, which further degrade solution quality as problem size grows. 
While TuRBO improves schedule quality up to \(N{=}8\), beyond this scale the combined impact of embedding and noise appears to dominate performance. 
To investigate the impact of noise on scalability limitations, we compare noise-free SQA with noisy simulations.
All subsequent experiments are conducted under TuRBO-based optimization of the annealing schedule.

\subsubsection*{Noise-free Simulation}
To investigate the noise effect, a useful next step would be to compare the hardware runs against simulated QA (SQA) on identical QUBO formulations~\cite{king2021scaling}. 
Such simulations provide a noise-free baseline that replicates the annealing dynamics without embedding-related constraints, allowing us to disentangle the algorithmic contribution of annealing schedules from hardware-induced errors~\cite{albash2018demonstration, bando2021simulated}. 

\subsubsection*{Noise Simulation}
\label{sec:noise-sim}

We model analog control errors in QA as perturbations of the problem Hamiltonian. 
Following D-Wave's ICE model, the QPU effectively solves a perturbed Ising Hamiltonian~\cite{yarkoni2022quantum, jeong2025embeddingawarenoisemodelingquantum}:
\begin{equation}
E^{\delta}(\boldsymbol{x}) = \sum_{i=1}^{N} \!\bigl(h_i+\delta h_i\bigr)x_i 
+ \sum_{i<j} \!\bigl(J_{ij}+\delta J_{ij}\bigr)x_i x_j,
\end{equation}
where $\delta h_i$ and $\delta J_{ij}$ denote deviations between programmed and realized coefficients. 
Such ICEs include several contributing factors, such as coefficient noise, persistent biases, and background susceptibility~\cite{dwave}. 
For simplicity, in this work we explicitly consider only two dominant manifestations: (i)~\textit{misspecification errors}, i.e., small perturbations of existing $h$ and $J$, and (ii)~\textit{ghost couplings} (background susceptibility), i.e., spurious next-nearest-neighbor terms that effectively introduce unintended interactions.
Accordingly, we model the total perturbation of couplers as
\begin{equation}
\delta J_{ij} = \delta J_{ij}^{\mathrm{misspec}} + \delta J_{ij}^{\mathrm{ghost}},
\end{equation}
while local fields are perturbed only through misspecification errors, $\delta h_i = \delta h_i^{\mathrm{misspec}}$.

\subsubsection*{Misspecification Error}
To emulate these hardware-scale imperfections while keeping problem costs interpretable, 
we inject relative Gaussian disorder on $h$ and $J$ with respect to the current BQM scale:
\begin{align}
\delta h_i^{\mathrm{misspec}} &\sim \mathcal{N}\!\bigl(0,\, \sigma_h\bigr), 
& \sigma_h &= \sigma_{h}^{\mathrm{rel}}\cdot \max_k |h_k|,\\
\delta J_{ij}^{\mathrm{misspec}} &\sim \mathcal{N}\!\bigl(0,\, \sigma_J\bigr), 
& \sigma_J &= \sigma_{J}^{\mathrm{rel}}\cdot \max_{u<v} |J_{uv}|.
\end{align}
Unless stated otherwise, we use $(\sigma_{h}^{\mathrm{rel}},\sigma_{J}^{\mathrm{rel}})=(0.05,\,0.02)$, 
corresponding to $5\%$ and $2\%$ of the dynamic ranges of $h$ and $J$, respectively~\cite{yarkoni2022quantum}. 

\subsubsection*{Ghost Couplings}
To capture background susceptibility and crosstalk, 
we randomly add couplings on a fraction $\rho_{\text{ghost}}$ of absent edges $(i,j)\notin\mathcal{E}$:
\begin{equation}
\delta J^{\mathrm{ghost}}_{ij} \sim \mathcal{N}\!\bigl(0,\, \sigma_{\text{ghost}}\bigr), 
\qquad \sigma_{\text{ghost}}=\sigma_{\text{ghost}}^{\mathrm{rel}}\cdot \max_{u<v}|J_{uv}|.
\end{equation}
We set $\rho_{\text{ghost}}=0.02$ and $\sigma_{\text{ghost}}^{\mathrm{rel}}=0.01$, 
thereby incorporating ICE background susceptibility into our SQA trials~\cite{dwave}.  

\begin{remark}
Our objective is to investigate schedule optimization under realistic analog uncertainty. 
We do not attempt to replicate device-specific error distributions; instead, we adopt a principled parametric proxy rooted in D-Wave’s ICE characterization 
(misspecification $\delta h,\delta J$ and ghost couplings) and in standard SQA methodology. 
\end{remark}

\begin{table}[t]
\centering
\caption{Comparison of SQA under noiseless and noisy settings, showing the gap to CPLEX optimal (\%), success probability, and best energy across different $N$.}
\label{tab:fair_comp1}
\begin{tabular}{c l c c c}
\hline
$N$ & Method & Gap (\%) & $p_{\text{succ}}$ & Best energy  \\
\hline
2 & SQA(Noiseless) & 0.00 & 1.0000 & 163.68 \\
  & SQA(Noise)     & 0.00 & 1.0000 &   163.68 \\
\hline
3 & SQA(Noiseless) & 0.00 & 1.0000 & 206.91 \\
  & SQA(Noise)     & 0.00 & 0.9977 &   206.91 \\
\hline
4 & SQA(Noiseless) & 0.00 & 1.0000 & 238.06 \\
  & SQA(Noise)     & 0.00 & 0.8938 &   238.06 \\
\hline
5 & SQA(Noiseless) & 0.00 & 1.0000 & 249.43 \\
  & SQA(Noise)     & 0.00 & 0.8070 &   249.43 \\
\hline
6 & SQA(Noiseless) & 0.00 & 1.0000 & 271.75 \\
  & SQA(Noise)     & 0.00 & 0.6978 &   271.75 \\
\hline
7 & SQA(Noiseless) & 0.00 & 1.0000 & 327.27 \\
  & SQA(Noise)     & 0.00 & 0.6035 &   327.27 \\
\hline
8 & SQA(Noiseless) & 0.00 & 1.0000 & 265.24 \\
  & SQA(Noise)     & 0.00 & 0.5562 &   265.24 \\
\hline
9 & SQA(Noiseless) & 0.00 & 1.0000 & 305.84 \\
  & SQA(Noise)     & 1.14 & 0.5180 &   309.33 \\
\hline
10 & SQA(Noiseless) & 0.00 & 1.0000 & 342.85 \\
  & SQA(Noise)     & 10.93 & 0.0558 &   380.33 \\
\hline
\end{tabular}
\end{table}

As shown in Table~\ref{tab:fair_comp1}, in all cases, noiseless SQA consistently achieved the optimal solution with zero optimality gap and unity success probability. 
In contrast, high performance in the noisy setting held only up to $N{=}8$. At $N{=}9$ the best energy was already suboptimal (gap $\approx 1.14\%$), and by $N{=}10$ the gap widened to $\approx 10.93\%$ with success probability nearly zero.

Figure~\ref{fig:dissresult} illustrates the difference in optimization dynamics. At $N{=}9$, noisy SQA exhibits fluctuations in energy across 100 iterations due to noise, while noiseless SQA remains stable near the optimum. At $N{=}10$, noisy SQA converges to a suboptimal solution with fluctuation, whereas noiseless SQA continues to locate the true ground state. The success probability plots further emphasize this divergence: noiseless SQA maintains $p_{\text{succ}}=1$, while noisy SQA decays with $N$ and nearly vanishes at $N{=}10$.

Overall, the comparative experiments between noiseless and noisy SQA reveal three key findings. 
First, both settings successfully reached optimal solutions up to $N{=}9$, but at $N{=}10$ noisy SQA failed to find the optimum and converged to suboptimal energies, whereas noiseless SQA continued to explore the ground state. 
Second, the success probability in the noiseless case remained consistently at~1 across all tested sizes, while in the noisy case it degraded steadily with $N$ and collapsed to nearly~0 at $N{=}10$. 
Third, the energy trajectories under TuRBO optimization (Fig.~\ref{fig:dissresult}) illustrate that noise induces significant fluctuations at $N{=}9$ and prevents convergence at $N{=}10$, in contrast to the stable convergence of the noiseless runs. In particular, the noisy SQA results highlight that noise strongly amplifies with increasing QUBO size, i.e., larger embedding sizes, which emerges as a primary scalability limitation. Hence, future directions should emphasize noise-reduction techniques that explicitly account for embedding overhead~\cite{raymond2025quantum, shingu2024quantum}.

\begin{figure}[!t]
\centering
\subfigure[Energy at $N{=}9$.]{%
    \includegraphics[width=0.24\textwidth]{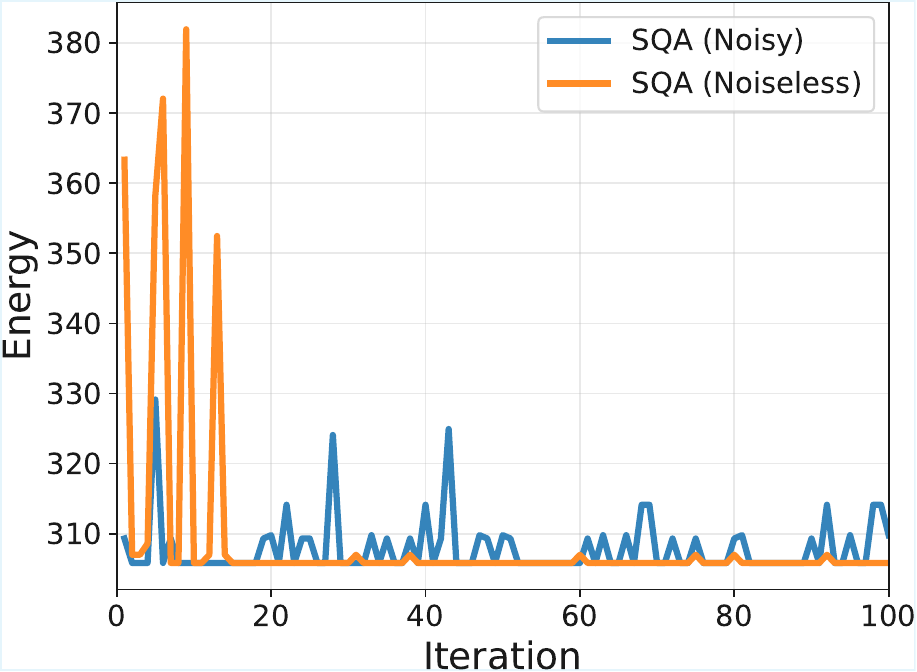}%
    \label{fig.N6}}
\subfigure[Energy at $N{=}10$.]{%
    \includegraphics[width=0.24\textwidth]{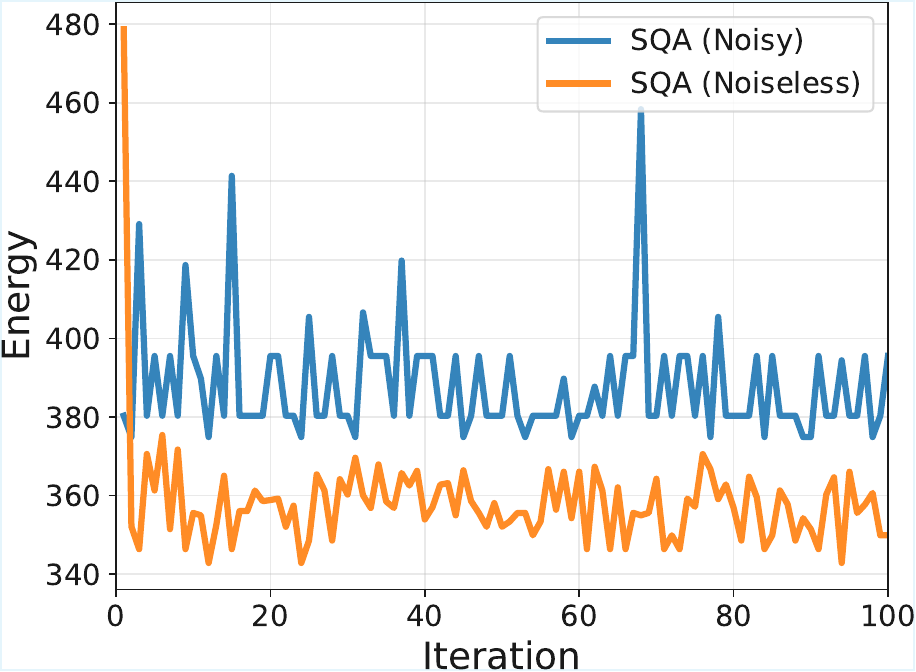}%
    \label{fig.N7}} \\[1ex]
\subfigure[Success probability at $N{=}9$.]{%
    \includegraphics[width=0.24\textwidth]{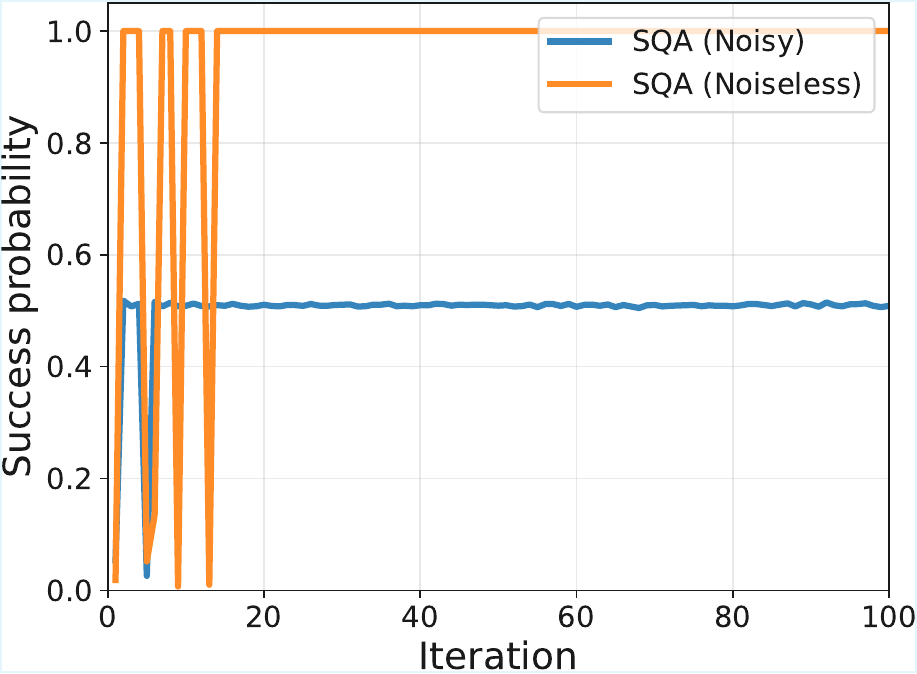}%
    \label{fig.N8}}
\subfigure[Success probability at $N{=}10$.]{%
    \includegraphics[width=0.24\textwidth]{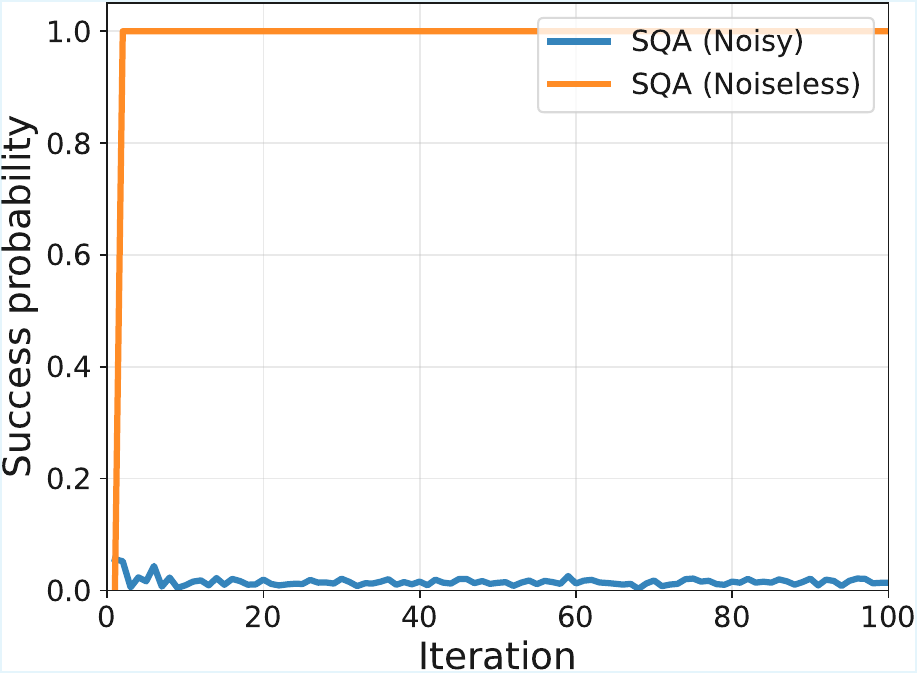}%
    \label{fig.N9}}
\caption{Comparison of optimization performance between noisy and noiseless SQA at $N{=}9$ and $N{=}10$. Subplots (a) and (b) show the evolution of energy across iterations, while (c) and (d) present the corresponding success probabilities.}
\label{fig:dissresult}
\end{figure}

\section{Conclusion}\label{sec:Con}
In this work, we proposed a TuRBO framework for annealing schedule design on a quantum annealer. By jointly optimizing annealing time and Fourier-parameterized schedules under fixed embeddings, our method consistently improved solution quality compared to random and greedy baselines. Through extensive experiments on TSP instances embedded on the D-Wave Advantage2 system, we demonstrated clear gains in energy quality, success probability, and chain break mitigation. Moreover, runtime analysis confirmed that properly tuned schedules enable QA to achieve competitive or superior TTS compared to classical baselines within the embedding-feasible range.
We also investigated scalability limits by contrasting noiseless and noisy SQA. The results reveal that while optimized schedules are highly effective up to moderate sizes ($N{\leq}7$), performance degrades beyond this point due to embedding overhead and hardware-induced noise. In particular, the collapse of success probability and the emergence of suboptimal convergence at $N{=}9$ highlight embedding-related error amplification as a key bottleneck.
Overall, our study bridges the gap between theoretical schedule optimization and real hardware deployment. It establishes TuRBO as a resource-efficient strategy for schedule design in NISQ regimes. Future work should focus on embedding-aware noise mitigation techniques and hybrid frameworks that combine optimized schedules with error correction or postprocessing, thereby extending scalability and robustness for larger real-world applications.

\bibliographystyle{IEEEtran}
\bibliography{bibliography}

\clearpage
\appendices

\section{Proof of Theorem~\ref{thm:turbo-complexity}}
\label{appendix:turbo-proof}

We analyze TuRBO in our setting where the design vector 
$\boldsymbol{s}=(T,\theta_1,\dots,\theta_M)$ specifies the Fourier schedule (Section~\ref{sec:met}) 
and the black-box objective is
\[
f(\boldsymbol{s}) \equiv E(\boldsymbol{s}) 
= \min_{r=1,\dots,R} H\!\big(\boldsymbol{x}^{(r)};\boldsymbol{s}\big).
\]
One QPU evaluation corresponds to executing the annealer with schedule $\boldsymbol{s}$, collecting $R$ number of reads, and returning $\min_r H(\boldsymbol{x}^{(r)};\boldsymbol{s})$.

\subsection{Per-iteration time and memory}

\begin{lemma}[Per-TR GP training cost]
\label{lem:gp-cost-app}
For a TR $j$ with $n_j \le n_{\max}$ points, training the GP surrogate requires 
Cholesky factorization of the $n_j \times n_j$ covariance matrix constructed with the 
Mat\'ern-5/2 kernel. 
This incurs $\mathcal{O}(n_j^3)$ time and $\mathcal{O}(n_j^2)$ memory, 
consistent with standard results on GP training complexity~\cite{giovanis2020data,eriksson2018scaling}.
\end{lemma}

\begin{proof}
The covariance matrix entries are defined by the Mat\'ern-5/2 kernel,
$K_{uv} = k_{\text{Mat\'ern 5/2}}(\boldsymbol{s}_u,\boldsymbol{s}_v)$, 
for all data points $\{\boldsymbol{s}_1,\dots,\boldsymbol{s}_{n_j}\}$ in the trust region. 
Cholesky factorization of this dense matrix requires $(1/3)n_j^3 + \mathcal{O}(n_j^2)$ floating-point operations, 
while storing the matrix and its factor requires $\mathcal{O}(n_j^2)$ memory. 
\end{proof}

\begin{lemma}[Acquisition maximization cost]
\label{lem:acq-cost-app}
Let $C_{\mathrm{acq}}(d,n_j)$ denote the expected cost of performing one local acquisition 
function maximization in dimension $d$ when the underlying GP surrogate in trust region $j$ 
is trained on $n_j$ points. Evaluating the acquisition function and its gradient requires 
posterior mean and variance computations, each of which costs $\mathcal{O}(n_j^2)$ using cached 
Cholesky factors~\cite{williams2006gaussian,shahriari2015taking}. If the optimizer performs 
$I(d)$ inner iterations on average, then
\[
C_{\mathrm{acq}}(d,n_j) = \mathcal{O}(I(d)\,n_j^2).
\]
With $R_{\mathrm{acq}}$ random restarts, the total acquisition maximization cost for TR $j$ is
\[
\mathcal{O}\!\big(R_{\mathrm{acq}}\,I(d)\,n_j^2\big).
\]
\end{lemma}

\begin{proof}
GP prediction at a new candidate requires computing the posterior mean and 
variance by solving triangular systems with the cached Cholesky factors of the covariance matrix. 
As noted in~\cite{williams2006gaussian,shahriari2015taking}, this costs $\mathcal{O}(n_j^2)$ per evaluation. 
A gradient-based optimizer in $d$ dimensions evaluates the acquisition and its gradient multiple 
times, with an expected $I(d)$ iterations. Thus the cost per restart is $\mathcal{O}(I(d)\,n_j^2)$, and 
multiplying by $R_{\mathrm{acq}}$ restarts gives the stated bound.
\end{proof}

\begin{proposition}[Per-iteration time and memory complexity]
\label{prop:per-iter-app}
Consider TuRBO with $M_{\mathrm{TR}}$ TRs, each storing at most $n_{\max}$ data points. 
In each iteration, every TR (i) refits its local GP surrogate and (ii) optimizes the acquisition function 
with $R_{\mathrm{acq}}$ random restarts in dimension $d$.

\begin{itemize}
    \item \textbf{GP training:} Refitting one GP costs $\mathcal{O}(n_j^3)$ time and $\mathcal{O}(n_j^2)$ memory for TR $j$.
    \item \textbf{Acquisition maximization:} Optimizing the acquisition costs 
    $\mathcal{O}(R_{\mathrm{acq}} I(d)\,n_j^2)$ time for TR $j$, 
    where $I(d)$ is the expected number of inner iterations of the gradient-based optimizer in dimension $d$.
\end{itemize}

Summing across all trust regions, the total per-iteration complexity is
\begin{equation}
\begin{aligned}
T_{\mathrm{iter}}
&= \mathcal{O}\left(\sum_{j=1}^{M_{\mathrm{TR}}} \left(n_j^3 + R_{\mathrm{acq}} I(d)\,n_j^2\right)\right), \\
S_{\mathrm{iter}}
&= \mathcal{O}\left(\sum_{j=1}^{M_{\mathrm{TR}}} n_j^2\right).
\end{aligned}
\end{equation}

If each $n_j \le n_{\max}$, this simplifies to
\begin{equation}
\begin{aligned}
T_{\mathrm{iter}}
&= \mathcal{O}\!\left(M_{\mathrm{TR}}\,n_{\max}^3 
   + R_{\mathrm{acq}} I(d)\,M_{\mathrm{TR}}\,n_{\max}^2\right), \\
S_{\mathrm{iter}}
&= \mathcal{O}\!\left(M_{\mathrm{TR}}\,n_{\max}^2\right).
\end{aligned}
\end{equation}

Here, $T_{\mathrm{iter}}$ denotes the per-iteration time complexity, 
while $S_{\mathrm{iter}}$ denotes the per-iteration memory complexity.
\end{proposition}

\subsection{Evaluation complexity to $\varepsilon$-stationarity}

To establish the overall complexity of TuRBO, we complement the per-iteration cost analysis 
with an evaluation complexity analysis that bounds the number of QPU evaluations required to 
reach an $\varepsilon$-stationary point. Our assumptions follow standard TR theory 
and GP BO analyses~\cite{eriksson2019scalable, davis2022gradient, Srinivas_2012, conn2000trust}:

\begin{itemize}
\item Local smoothness: $E(\boldsymbol{s})$ is continuously differentiable with an $L$-Lipschitz gradient.
\item Model accuracy: With probability $1-\delta$, 
$|E(\boldsymbol{s})-\mu_t(\boldsymbol{s})|\le \sqrt{\beta_t}\,\sigma_t(\boldsymbol{s})$ for all $\boldsymbol{s}$ in the TR,
as in GP upper confidence bound analyses~\cite{Srinivas_2012}.
\item Sufficient decrease: Acquisition maximization yields an expected decrease proportional 
to the Cauchy step whenever $\|\nabla E(\boldsymbol{s}_t)\|>\varepsilon$, 
consistent with standard trust-region sufficient decrease conditions~\cite{conn2000trust,eriksson2019scalable}.
\end{itemize}

\begin{proposition}[Evaluation complexity]
\label{prop:eval-app}
Under (A1)-(A3), TuRBO (single-point) finds $\boldsymbol{s}$ with $\|\nabla E(\boldsymbol{s})\|\le \varepsilon$ within
\begin{equation}
\mathcal{O}\!\big(\varepsilon^{-2} + \log(\ell_0/\ell_{\min})\big)
\end{equation}
QPU evaluations, up to logarithmic factors in $1/\varepsilon$.
\end{proposition}

\begin{proof}
Trust-region theory shows each successful step reduces $E$ by at least $\mathcal{O}(\varepsilon^2/L)$ until 
$\|\nabla E\|\le \varepsilon$. Thus $\mathcal{O}(1/\varepsilon^2)$ successes suffice. The number of failed steps 
before shrinkage is bounded in expectation; total shrinks is $\mathcal{O}(\log(\ell_0/\ell_{\min}))$. 
Each step uses one evaluation. Extra $\log(1/\varepsilon)$ factors come from model-confidence terms $\beta_t$.
\end{proof}

Here, $\ell_0$ denotes the initial TR radius and $\ell_{\min}$ the minimum TR radius. 
The term $\log(\ell_0/\ell_{\min})$ accounts for the maximum number of shrink operations when the TR is repeatedly reduced from its initial to its minimum allowable size.

Proposition~\ref{prop:per-iter-app} gives the per-iteration complexity, while Proposition~\ref{prop:eval-app} gives the evaluation complexity, together establishing Theorem~\ref{thm:turbo-complexity}.

\section{Scailability for Clique Embeddings on Zephyr Topology}\label{app:chain-scaling}
\begin{table*}[t]
\centering
\caption{Notation used for clique embedding on Zephyr topology.}
\label{tab:notation}
\begin{tabular}{c l}
\hline
Symbol & Description \\
\hline
$N$ & Problem size (e.g., number of customer cities in TSP) \\
$L$ & Number of logical variables (QUBO size); for TSP, $L=(N+1)^2$ \\
$K_L$ & Complete graph on $L$ vertices (logical interaction graph) \\
$V_{\mathrm{phys}}$ & Set of physical qubits in Zephyr topology \\
$Q = |V_{\mathrm{phys}}|$ & Total number of physical qubits available on hardware (capacity) \\
$\Delta(G_{\mathrm{Z}})$ & Maximum degree of Zephyr hardware graph ($=20$ for $t=4$) \\
$\phi(v)$ & Chain (connected subgraph of $V_{\mathrm{phys}}$) representing logical variable $v$ \\
$\ell_v = |\phi(v)|$ & Chain length of logical variable $v$ \\
$\ell(L) = \tfrac{1}{L}\sum_{v=1}^L \ell_v$ & Mean chain length for clique embedding of $K_L$ \\
$Q_{\mathrm{phys}}(L)$ & Total number of physical qubits used in embedding \\
& $Q_{\mathrm{phys}}(L) = \sum_{v=1}^L \ell_v = L\cdot \ell(L)$ \\
$w(G)$ & Bisection width (minimum number of crossing edges in balanced cut of graph $G$) \\
\hline
\end{tabular}
\end{table*}
For clarity, all notation used in this section is summarized in Table~\ref{tab:notation}.
The Zephyr graph $Z_m$ with tile parameter $t$ and grid parameter $m$ contains 
\begin{equation}
|V_{\mathrm{phys}}| = 8tm^2 + 4tm
\end{equation}
physical qubits in total (the hardware capacity), 
and
\begin{equation}
\Delta(G_{\mathrm{Z}}) = 4t+4,
\end{equation}
which is the maximum degree (the largest number of couplers incident to a single qubit). 
For the standard case $t=4$, this gives $\Delta(G_{\mathrm{Z}})=20$~\cite{Boothby2021Zephyr}.

When embedding a logical clique $K_L$ into $Z_m$, 
each logical variable $v\in V_{\mathrm{log}}$ is mapped to a connected subgraph $\phi(v)\subseteq V_{\mathrm{phys}}$, 
called a \textit{chain}. Denote the chain length of $v$ as $\ell_v=|\phi(v)|$ and the mean chain length as
\begin{equation}
\ell(L) = \frac{1}{L}\sum_{v=1}^L \ell_v.
\end{equation}
The total number of physical qubits used by the embedding is
\begin{equation}
Q_{\mathrm{phys}}(L) = \sum_{v=1}^L \ell_v = L \cdot \ell(L).
\label{eq:qphys-def}
\end{equation}

\subsection{Lower Bound for Chain Length}

The bisection width of the complete graph $K_L$ is
\begin{equation}
w(K_L) = \Theta(L^2),
\end{equation}
meaning that any balanced cut separates $\Theta(L^2)$ edges.  
On the other hand, Zephyr is a sparse quasi-2D bounded-degree topology. 
As a consequence of separator theorems for grid-like graphs~\cite{tarjan1985decomposition}, any balanced separator in $Z_m$ has capacity
\begin{equation}
w(Z_m) = \mathcal{O}(\sqrt{|V_{\mathrm{phys}}|}) = \mathcal{O}(\sqrt{Q}).
\end{equation}
To realize $\Theta(L^2)$ crossing edges of $K_L$ across a separator of capacity $\mathcal{O}(\sqrt{Q})$, the average path length required is at least~\cite{leighton1999multicommodity}
\begin{equation}
\Omega\!\left(\frac{L^2}{\sqrt{Q}}\right).
\end{equation}
With bounded degree, each logical chain must span this routing radius. 
Hence the mean chain length satisfies
\begin{equation}
\ell(L) \;\ge\; c_1 \frac{L^2}{\sqrt{Q}}
\end{equation}
for some constant $c_1>0$.  
Using \eqref{eq:qphys-def}, i.e., $Q=\Theta(L\cdot \ell(L))$, we obtain
\begin{equation}
\ell(L)^{3/2} \;\ge\; c_1 L^{3/2}.
\end{equation}
Taking both sides to the power of $2/3$ yields
\begin{equation}
\ell(L) \;\ge\; c_2 \sqrt{L},
\end{equation}
where $c_2 = c_1^{2/3} > 0$. 
Thus we conclude that $\ell(L)=\Omega(\sqrt{L})$.

\subsection{Upper Bound for Chain Length}

We give a constructive upper bound on the mean chain length by using the explicit clique embeddings provided for the Zephyr topology~\cite{Boothby2021Zephyr}. 
Let $Z_{m,t}$ denote a Zephyr graph with grid parameter $m$ and tile parameter $t$ (the standard case uses $t{=}4$). 
The report shows two constructive embedding templates:
\begin{itemize}
\item Standard Clique Embedding: $K_{16m-8}$ can be embedded in $Z_{m,4}$ with \emph{chain length} exactly $m$;
\item Extended Clique Embedding: $K_{16m+1}$ can be embedded in $Z_{m,4}$ with \emph{chain length} at most $2m$.
\end{itemize}
(See \cite{Boothby2021Zephyr}, Sections~2.2 and~5.1 for the explicit constructions.)

Let $L$ denote the clique size to be embedded (i.e., the number of logical variables). 
From the standard clique embedding, if $L\le 16m{-}8$, then $m \ge \frac{L+8}{16}$ and the chain length per logical variable is $m$. 
From the extended clique embedding, if $L\le 16m{+}1$, then $m \ge \frac{L-1}{16}$ and the chain length per logical variable is at most $2m$.
Therefore, for all feasible $L$, we have the uniform (constructive) upper bound
\begin{equation}
\ell(L) \;\le\; 2m \;\le\; \frac{L+8}{8} \;=\; \mathcal{O}(L),
\label{eq:clique-upper-L}
\end{equation}
and in particular $\ell(L)\le m$ when the standard clique embedding applies. Hence, the mean chain length scales \emph{linearly} in the clique size $L$ on Zephyr.

\subsection{Equivalent form in terms of hardware size}
The number of physical qubits of $Z_{m,t}$ is $|V_{\mathrm{phys}}| = 8tm^2 + 4tm = \Theta(m^2)$ (with maximum degree $4t{+}4$, i.e., $20$ for $t{=}4$). 
Thus $m=\Theta(\sqrt{Q})$ where $Q:=|V_{\mathrm{phys}}|$. 
Under the clique constructions above, both the chain length and the maximum embeddable clique size scale as $m$, i.e.,
\begin{equation}
\ell(L) \;=\; \Theta(\sqrt{Q}),\qquad L \;=\; \Theta(\sqrt{Q}).
\label{eq:clique-upper-Q}
\end{equation}
Eliminating $\sqrt{Q}$ between the two relations recovers \eqref{eq:clique-upper-L}, i.e., $\ell(L)=\Theta(L)$ in the capacity-saturating regime.

\begin{remark}[Tightness for Chain Length]
The separator-based lower bound shows $\ell(L)\!=\!\Omega\!\big(L^2/\sqrt{Q}\big)$; 
combining with $Q{=}\Theta(m^2)$ and $L{=}\Theta(m)$ on Zephyr yields $\ell(L){=}\Omega(m){=}\Omega(L)$ (hence also $\Omega(\sqrt{L})$), 
which matches the constructive upper bound \eqref{eq:clique-upper-L} up to constants. 
Therefore, for clique embeddings on Zephyr, the mean chain length scaling is tight: $\ell(L)=\Theta(L)$ as a function of $L$, and equivalently $\ell(L)=\Theta(\sqrt{Q})$ as a function of the hardware size $Q$.    
\end{remark}

\subsubsection*{Experiments Analysis}
Table~\ref{tab:clique-embed} empirically illustrates the clique embedding behavior for TSP QUBOs on Zephyr. 
Here $N$ denotes the number of customer cities, leading to $L=(N{+}1)^2$ logical variables in the assignment formulation, 
with each logical variable represented by a chain of physical qubits. 
The table reports three quantities: the number of logical variables $L$, the number of physical qubits used $Q_{\mathrm{phys}}(L)$, 
and the mean chain length $\ell(L)$. 

The data confirm that both $\ell(L)$ and $Q_{\mathrm{phys}}(L)$ grow sharply with $N$. 
For instance, the mean chain length increases from $1.0$ at $N{=}1$ ($L=4$) to $11.5$ at $N{=}9$ ($L=100$), 
while the physical-qubit footprint rises from $4$ to $1151$ over the same range. 
These values fall between the theoretical bounds: the separator-based lower bound $\ell(L)=\Omega(\sqrt{L})$ 
and the constructive upper bound $\ell(L)=\Theta(L)$. 
At small $N$, the growth of $\ell(L)$ is closer to the $\sqrt{L}$ scaling, but as $N$ increases it approaches the linear regime. 

Consequently, the total embedding cost $Q_{\mathrm{phys}}(L)=L\cdot\ell(L)$ scales between $\Theta(L^{3/2})$ and $\Theta(L^2)$. 
The observed values show that clique embeddings on Zephyr rapidly consume physical qubits: 
already at $N=9$, more than one quarter of the available 4800 qubits are required. 
This sharp growth highlights the scalability bottleneck of clique embeddings: as $N$ increases, longer chains amplify hardware noise and chain-break probability, which directly degrades solution quality in practice.

\end{document}